\documentclass[manyauthors]{fundam}
\usepackage{mysty}

\begin{document}

\title{Büchi-Elgot-\!Trakhtenbrot Theorem for Higher- Dimensional Automata}

\author{%
  Amazigh Amrane \\ EPITA Research Laboratory (LRE), France
  \and
  Hugo Bazille \\ EPITA Research Laboratory (LRE), France
  \and
  Emily Clement \\ CNRS, LIPN UMR 7030, Université Sorbonne Paris Nord,
  France
  \vspace*{0.5cm}
  \and
  Uli Fahrenberg \\ EPITA Research Laboratory (LRE), France
  \and
  Marie Fortin \\ Université Paris Cité, CNRS, IRIF,
  France
  \and
  Krzysztof Ziemiański \\ University of Warsaw, Poland
}

\maketitle

\runninghead{A.~Amrane et al}{Büchi-Elgot-\!Trakhtenbrot Theorem for Higher-Dimensional Automata}

\begin{abstract}
In this paper we explore languages of higher-dimensional automata
(HDAs) from an algebraic and logical point of view. Such languages are sets of
finite width-bounded interval pomsets with interfaces (ipomsets) closed
under order extension.
We show that ipomsets can be represented as equivalence classes of words over a particular alphabet, called step sequences.
We introduce an automaton model that recognize such languages.
Doing so allows us to lift the classical Büchi-Elgot-\!Trakhtenbrot Theorem to languages of HDAs: we prove that a set of interval ipomsets is the language of an HDA if and only if it is simultaneously MSO-definable, of bounded width, and closed under order refinement.

\end{abstract}


\section{Introduction}

Connections between logic and automata play a key role in several areas of 
theoretical computer science -- logic being used to specify the behaviours of
automata models in formal verification, and  automata being used to prove the 
decidability of various logics.
The first and most well-known result of this kind is the equivalence in expressive
power of finite automata and monadic second-order logic (MSO) over finite 
words, proved independently by B\"uchi \cite{Buechi60}, Elgot \cite{Elgot1961} and
Trakhtenbrot \cite{Trakhtenbrot62} in the 60's.
This was soon extended to infinite words \cite{Buchi1962} as well as finite and
infinite trees \cite{ThatcherW68,Rabin1969,Doner70}.

Finite automata over words are a simple model of sequential systems with a finite
memory, each word accepted by the automaton corresponding to an execution
of the system.
For concurrent systems, executions may be represented as \emph{pomsets} (partially ordered multisets or, equivalently, labelled partially ordered sets). 
Several classes of pomsets and matching automata models have been defined
in the literature, corresponding to different communication models or different
views of concurrency.
In that setting, logical characterisations of classes of automata in the spirit of the
B\"uchi-Elgot-Trakhtenbrot theorem have been obtained for several cases, such as
asynchronous automata and Mazurkiewicz traces \cite{Zielonka87,tho90traces}, 
branching automata and series-parallel pomsets \cite{Kuske00, Bedon15}, step transition systems  and local trace languages \cite{DBLP:conf/apn/FanchonM09,DBLP:conf/concur/KuskeM00},
or communicating finite-state machines and message sequence charts
\cite{GKM06}.

\emph{Higher-dimensional automata} (\emph{HDAs}) \cite{Pratt91-geometry,Glabbeek91-hda} 
are another automaton-based model of concurrent systems.
Initially studied from a geometrical or categorical point of view, the language
theory of HDAs has become another focus for research in the past few years
\cite{DBLP:journals/mscs/FahrenbergJSZ21}.
Languages of HDAs are sets of \emph{interval pomsets with
  interfaces} (\emph{ipomsets})~\cite{DBLP:journals/iandc/FahrenbergJSZ22}.
The idea is that each event in the execution of an HDA corresponds to an
interval of time where some process is active.

\begin{figure}[tbp]
	\centering
	\begin{tikzpicture}[x=1cm]
		\def\possh{-1.3}
		\begin{scope}[shift={(8,0)}]
			\def\hw{0.3}
			\draw[-](0,0)--(0,1.7);
			\fill[fill=green!50!white,-](-0.025,1.2)--(1.2,1.2)--(1.2,1.2+\hw)--(-0.025,1.2+\hw);
			\draw[-] (0,1.2)--(1.2,1.2)--(1.2,1.2+\hw)--(0,1.2+\hw);
			\filldraw[fill=pink!50!white,-](1.3,0.7)--(1.9,0.7)--(1.9,0.7+\hw)--(1.3,0.7+\hw)--(1.3,0.7);
			\filldraw[fill=blue!20!white,-](0.5,0.2)--(1.7,0.2)--(1.7,0.2+\hw)--(0.5,0.2+\hw)--(0.5,0.2);
			\draw[-](2.2,0)--(2.2,1.7);
			\node at (0.6,1.2+\hw*0.5) {$a$};
			\node at (1.6,0.7+\hw*0.5) {$b$};
			\node at (1.1,0.2+\hw*0.5) {$c$};
		\end{scope}
		\begin{scope}[shift={(8,\possh)}]
			\node (a) at (0.4,0.7) {$\ibullet a$};
			\node (c) at (0.4,-0.7) {$c$};
			\node (b) at (1.8,0) {$b$};
			\path (a) edge (b);
			\path[densely dashed, gray]  (b) edge (c) (a) edge (c);
		\end{scope}
		\begin{scope}[shift={(9,2.1*\possh)}]
			\node (a) at (0,0) {$%
				\loset{\ibullet a \ibullet \\ \pibullet c \ibullet}
				\loset{\ibullet a \pibullet \\ \ibullet c \ibullet}
				\loset{\pibullet b \ibullet \\ \ibullet c \ibullet}
				\loset{\ibullet b \\ \ibullet c}$};
		\end{scope}
		\begin{scope}[shift={(4,0)}]
			\def\hw{0.3}
			\draw[-](0,0)--(0,1.7);
			\fill[fill=green!50!white,-](-0.025,1.2)--(1.2,1.2)--(1.2,1.2+\hw)--(-0.025,1.2+\hw);
			\draw[-] (0,1.2)--(1.2,1.2)--(1.2,1.2+\hw)--(0,1.2+\hw);
			\filldraw[fill=pink!50!white,-](1.3,0.7)--(1.9,0.7)--(1.9,0.7+\hw)--(1.3,0.7+\hw)--(1.3,0.7);
			\filldraw[fill=blue!20!white,-](0.5,0.2)--(1.1,0.2)--(1.1,0.2+\hw)--(0.5,0.2+\hw)--(0.5,0.2);
			\draw[-](2.2,0)--(2.2,1.7);
			\node at (0.6,1.2+\hw*0.5) {$a$};
			\node at (1.6,0.7+\hw*0.5) {$b$};
			\node at (0.8,0.2+\hw*0.5) {$c$};
		\end{scope}
		\begin{scope}[shift={(4,\possh)}]
			\node (a) at (0.4,0.7) {$\ibullet a$};
			\node (c) at (0.4,-0.7) {$c$};
			\node (b) at (1.8,0) {$b$};
			\path (a) edge (b) (c) edge (b);
			\path[densely dashed, gray]  (a) edge (c);
		\end{scope}
		\begin{scope}[shift={(5,2.1*\possh)}]
			\node (a) at (0,0) {$%
				\loset{\ibullet a \ibullet \\ \pibullet c \ibullet}
				\loset{\ibullet a \\ \ibullet c}
				b\ibullet\, \ibullet b$};
		\end{scope}
		\begin{scope}[shift={(0.0,0)}]
			\def\hw{0.3}
			\draw[-](0,0)--(0,1.7);
			\fill[fill=green!50!white,-](-0.025,1.2)--(0.4,1.2)--(0.4,1.2+\hw)--(-0.025,1.2+\hw);
			\draw[-](0,1.2)--(0.4,1.2)--(0.4,1.2+\hw)--(0,1.2+\hw);
			\filldraw[fill=pink!50!white,-](1.3,0.7)--(1.9,0.7)--(1.9,0.7+\hw)--(1.3,0.7+\hw)--(1.3,0.7);
			\filldraw[fill=blue!20!white,-](0.5,0.2)--(1.1,0.2)--(1.1,0.2+\hw)--(0.5,0.2+\hw)--(0.5,0.2);
			\draw[-](2.2,0)--(2.2,1.7);
			\node at (0.2,1.2+\hw*0.5) {$a$};
			\node at (1.6,0.7+\hw*0.5) {$b$};
			\node at (0.8,0.2+\hw*0.5) {$c$};
		\end{scope}
		\begin{scope}[shift={(0.0,\possh)}]
			\node (a) at (0.4,0.7) {$\ibullet a$};
			\node (c) at (0.4,-0.7) {$c$};
			\node (b) at (1.8,0) {$b$};
			\path (a) edge (b) (c) edge (b) (a) edge (c);
		\end{scope}
		\begin{scope}[shift={(1,2.1*\possh)}]
			\node (a) at (0,0) {$\ibullet a\, c\ibullet\, \ibullet c\, b\ibullet\, \ibullet b$};
		\end{scope}
	\end{tikzpicture}
	\caption{Activity intervals of events (top),
		corresponding ipomsets (middle),
		and representation as step sequences (bottom).
		Full arrows indicate precedence order;
		dashed arrows indicate event order;
		bullets indicate interfaces.}
	\label{fi:iposets1}
\end{figure}
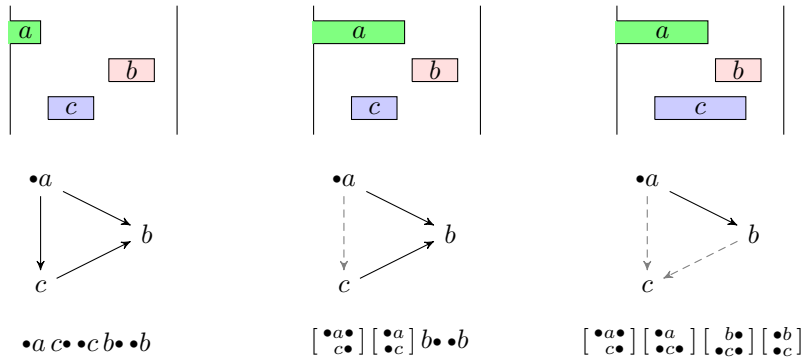

Examples with three activity intervals labelled $a$, $b$, and $c$ are shown in the top of Figure~\ref{fi:iposets1}.
These events are then partially ordered as follows: two events are ordered if
the first one ends before the second one starts, and they are concurrent if
they overlap.
This gives rise to a pomset as shown in the middle of Figure~\ref{fi:iposets1}.
We allow some events
to be started before the beginning (this is the case for the $a$-labelled events
in Figure~\ref{fi:iposets1}), and some events might never be
terminated. Such events define the \emph{interfaces}:
events which are already active in the beginning form the source interface,
those continuing beyond the end, the target interface.


Interval-order pomsets as introduced in \cite{%
  journals/mpsy/Fishburn70,
  book/Fishburn85},
with or without interfaces, are important in other areas of concurrency theory \cite{%
  DBLP:journals/iandc/JanickiY17,
  DBLP:series/sci/2022-1020,
  DBLP:journals/fuin/JanickiK19}
and distributed computing \cite{%
  DBLP:journals/cacm/Lamport78,
  DBLP:journals/dc/Lamport86,
  DBLP:journals/jacm/CastanedaRR18},
as well as relativity theory \cite{Wiener14}.
Compared to, for example, series-parallel pomsets \cite{%
  DBLP:journals/tcs/Gischer88,
  DBLP:journals/tcs/BloomE96a,
  DBLP:journals/mscs/BloomE97},
their algebraic theory is, however, much less developed.
Based on the antichain representations of \cite{%
  book/Fishburn85,
  DBLP:journals/fuin/JanickiK19}
and picking up on ideas in \cite{%
  DBLP:journals/fuin/FahrenbergZ24,
  AMRANE2025115156,
  conf/apn/AmraneBCF24},
we develop here the algebraic theory of interval ipomsets.

We prove that the category of interval ipomsets is isomorphic to one of \emph{step sequences},
which are equivalence classes of words on special discrete ipomsets under a natural relation.
The bottom line of Figure \ref{fi:iposets1} shows the step sequences corresponding to the respective ipomsets.
We also introduce an automaton model for such step sequences, called \emph{ST-automata}
and based on work in \cite{%
  DBLP:journals/lmcs/FahrenbergJSZ24,
  AMRANE2025115156,
  conf/apn/AmraneBCF24,
  DBLP:journals/lites/Fahrenberg22,
  DBLP:conf/adhs/Fahrenberg18},
and show that any HDA may be translated to an ST-automaton with equivalent language.

If we shorten some intervals in an interval representation as in Figure~\ref{fi:iposets1},
then some events which were concurrent become ordered.
Such introduction of precedence order (in Figure~\ref{fi:iposets1}, from right to left)
is called order refinement, and its inverse
(removing precedence order by prolonging intervals; the left-to-right direction in Figure~\ref{fi:iposets1})
is called \emph{subsumption}.
These notions are important in the theory of HDAs, as their languages are closed under subsumption.
We also develop the algebraic theory of subsumptions,
using elementary subsumptions on step sequences.

Several theorems of classical automata theory have already been ported
to higher-dimensional automata, including a Kleene theorem \cite{DBLP:journals/lmcs/FahrenbergJSZ24}
and a Myhill-Nerode theorem~\cite{DBLP:journals/fuin/FahrenbergZ24}.
Here we extend this basic theory of HDAs by exploring the relationship between HDAs and monadic second-order logic.
We prove that a set of interval ipomsets is the language of an HDA if and only if
it is simultaneously MSO-definable, of bounded width, and closed under subsumption.

To do so,
we extend the correspondence between step sequences and interval ipomsets
to the logic side, by showing that a language of interval ipomsets is MSO-definable if and only if the corresponding language of step sequences is
MSO-definable. More specifically, we construct an MSO interpretation of
interval ipomsets into step sequences
(or rather into their representatives), and of
canonical representatives of step sequences into interval ipomsets.
We then use these translations
in order to leverage the classical B\"uchi-Elgot-Trakhtenbrot over words. To go from the language of an HDA to a regular language of words, we also rely on our translation from HDAs to ST-automata. In the other direction, we go through rational expressions and make use of the Kleene theorem for HDAs \cite{DBLP:journals/lmcs/FahrenbergJSZ24}.

Preliminary versions of these results have been presented at DLT 2024 \cite{conf/dlt/AmraneBFF24} and RAMICS 2024 \cite{ABCFZ24}. Among other changes, here we simplify the presentation of the B\"uchi-Elgot-Trakhtenbrot theorem from \cite{conf/dlt/AmraneBFF24} by making explicit the relation between MSO over interval ipomsets and MSO over ST-sequences (Theorem~\ref{th:MSOeq}), as well as relying on the ST-automata from \cite{ABCFZ24}.
We also correct an error in Section~3 of \cite{ABCFZ24}.
There, Lemma 16 takes a subsumption and splits off an elementary subsumption
which removes precisely one pair from the precedence order;
but naive application of the lemma does not ensure that the resulting intermediate pomset is an interval order.
We work around this problem by proving what is now Lemma \ref{le:subselm-complete}
(Lemma 24 in \cite{ABCFZ24}) without making use of the above lemma.

The paper is organised as follows.
Pomsets with interfaces, interval orders and step decompositions  are defined in Section~\ref{se:ipomsets-step}. 
Two isomorphic categories of interval pomsets with interfaces are defined in Section~\ref{se:catiipoms} along with an isomorphic representation based on decompositions called step sequences.
We characterize ipomset subsumptions in Section~\ref{se:step-subsumptions} through their step sequence representations.
In Section~\ref{se:hda-st-a}, we transfer the isomorphisms from Section~\ref{se:catiipoms} to the operational setting by introducing a translation from HDAs to a class of state-labeled automata, called ST-automata, that preserves languages up to these isomorphisms.
Finally, we explore ipomsets and step sequences from a logical point of view in Section~\ref{sec:MSO} and show that, again up to isomorphism, monadic second-order logic has the same expressive power over step sequences and ipomsets when a width bound is fixed.

\section{Pomsets with Interfaces}
\label{se:ipomsets-step}
\label{sse:ipomsets}


In this section we introduce pomsets with interfaces (ipomsets),
which play the role of words for higher-dimensional automata.
We also recall interval orders
and develop the fact that an ipomset is interval if and only if its admits a step decomposition.
We fix an alphabet $\Sigma$, finite or infinite, throughout this paper.

We first define concrete ipomsets, which are finite labelled sets ordered by two relations.
Later on we will take isomorphism classes and call these \emph{ipomsets}.
The two relations on concrete ipomsets are the \emph{precedence} order $<$,
which is a strict partial order
(\ie asymmetric, transitive and thus irreflexive)
and denotes temporal precedence of events,
and the \emph{event} order,
which is an acyclic relation
(\ie such that its transitive closure is a strict partial order)
that restricts to a total order on each $<$-antichain.
The latter is needed to distinguish concurrent events, in particular in the presence of autoconcurrency, and may be seen as a form of process identity.
Concrete ipomsets also contain \emph{sources} and \emph{targets}
which we will use later to define the \emph{gluing} of such structures.

\begin{definition}
  \label{de:ipoms}
  A \emph{concrete ipomset} (over $\Sigma$) is a structure
  $(P, {<}, {\evord}, S, T, \lambda)$
  consisting of the following:
  \begin{itemize}
  \item a finite set $P$ of \emph{events};
  \item a strict partial order $<$ on $P$ called the \emph{precedence order};
  \item an acyclic relation ${\evord}\subseteq P\times P$ called the \emph{event order};
  \item a subset $S\subseteq P$ called the \emph{source set};
  \item a subset $T\subseteq P$ called the \emph{target set}; and
  \item a \emph{labeling} $\lambda: P\to \Sigma$.
  \end{itemize}
  We require that
  \begin{itemize}
  \item for every $x,y\in P$ exactly one of the following holds:
  \[x<y,\qquad y<x,\qquad x\evord y,\qquad y\evord x,\qquad x=y;\]
  \item events in $S$ are $<$-minimal in $P$, \ie for all $x\in S$ and $y\in P$, $y\not< x$; and
  \item events in $T$ are $<$-maximal in $P$, \ie for all $x\in T$ and $y\in P$, $x\not< y$.
  \end{itemize}
\end{definition}

We may add subscripts ``${}_P$'' to the elements above if necessary
and omit any empty substructures from the signature.
We will also often use the notation $\ilo{S}{P}{T}$ instead of $(P, {<}, {\evord}, S, T, \lambda)$
if no confusion may arise.

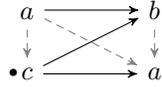
\begin{figure}[tbp]
  \centering
  \begin{tikzpicture}[x=1.2cm, y=1.2cm]
    \node (a) at (0.4,1.4) {$a\vphantom{d}$};
    \node (c) at (0.4,0.7) {$c$};
    \node at (0.25,0.7) {$\ibullet$};
    \node (b) at (1.8,1.4) {$b\vphantom{d}$};
    \node (d) at (1.8,0.7) {$a$};
    \path (a) edge (b) (c) edge (d) (c) edge (b);
    \path[densely dashed, gray] (b) edge (d) (a) edge (d) (a) edge (c);
  \end{tikzpicture}
  \caption{An ipomset, \cf Example~\ref{ex:n}.}
  \label{fi:n}
\end{figure}

\begin{remark}
  This definition of an ipomset is slightly different but equivalent 
  to the definitions given in 
  \cite{DBLP:journals/mscs/FahrenbergJSZ21, DBLP:journals/lmcs/FahrenbergJSZ24}.
  Here we drop the transitivity condition for the event order,
  which allows for a more natural notion of isomorphism.
  We will get back to this issue in Section \ref{se:ipoms-trans}.
\end{remark}

\begin{example}
  \label{ex:n}
  Figure~\ref{fi:n} depicts a concrete ipomset $P=\{x_1, x_2, x_3, x_4\}$ with four events
  labelled by $\lambda(x_1)= \lambda(x_4)=a$, $\lambda(x_2)=b$ and $\lambda(x_3)=c$.
  (We do not show the names of events, only their labels.)
  Its precedence order is given by $x_1<x_2$, $x_3<x_2$ and $x_3<x_4$
  and its event order by $x_1\evord x_3$, $x_1\evord x_4$ and $x_2\evord x_4$.
  The sources are $S=\{x_3\}$ and the targets $T=\emptyset$.
  (We denote these using ``$\ibullet$''.)
  We think of events in $S$ as being already active at the beginning of $P$,
  and the ones in $T$ (here there are none) continue beyond the ipomset~$P$.
\end{example}

A bijection $f:P\to Q$ between concrete ipomsets is an \emph{isomorphism}
if it preserves and reflects the structure;
that is,
\begin{itemize}
\item $f(S_P)=S_Q$, $f(T_P)=T_Q$, $\lambda_Q\circ f=\lambda_P$,
\item $f(x)<_Q f(y)$ iff $x<_P y$, and $x\evord_P y$ iff $f(x)\evord_Q f(y)$.
\end{itemize}
We write $P\simeq Q$ if $P$ and $Q$ are isomorphic.
The following result extends \cite[Lem.\@ 34]{DBLP:journals/mscs/FahrenbergJSZ21}. 

\begin{lemma}
\label{le:UniqueIso}
	There is at most one isomorphism between any two concrete ipomsets.
\end{lemma}

\begin{proof}
	Let $f,g:P\to Q$ be isomorphisms of concrete ipomsets.
	By induction on $n=|P|=|Q|$ we will show that $f=g$. 
	For $n=1$ this is obvious,
	so assume that $n>1$. 
	The set of $<$-minimal elements $P_{\min}$ of $P$
	is totally $\evord$-ordered and thus
	there is a unique $\evord$-minimal element $x_0\in P_{\min}$.
	Similarly, there is a unique $\evord$-minimal element $y_0\in Q_{\min}$.
	Thus, $f(x_0)=g(x_0)=y_0$.
	The restrictions $f\rest {P\setminus\{x_0\}}$ and $g\rest {P\setminus\{x_0\}}$
	are isomorphisms of concrete ipomsets
	$P\setminus \{x_0\}\to Q\setminus\{y_0\}$
	of cardinality $n-1$ and hence they are equal.
	Eventually, $f=g$.
\end{proof}

\begin{definition}
	An \emph{ipomset} is an isomorphism class of concrete ipomsets.
\end{definition}

Thanks to Lemma \ref{le:UniqueIso} we 
can switch freely between ipomsets and their concrete representatives,
which we will do without further notice whenever convenient.
Furthermore, we may always choose representatives in isomorphism classes
such that isomorphisms become equalities:

\begin{lemma}
  \label{le:isoeq}
  For any concrete ipomsets $P$ and $Q$ with $T_P\simeq S_Q$ there exists $Q'\simeq Q$
  such that $T_P=S_{Q'}=P\cap Q'$.
\end{lemma}

\begin{proof}
  Write $P=\{p_1,\dotsc, p_n\}$ such that $T_P=\{p_1,\dotsc, p_k\}$ for some $k\ge 1$.
  (The lemma is trivially true for $T_P=\emptyset$.)
  Let $f: T_P\to S_Q$ be the (unique) isomorphism
  and enumerate $Q=\{q_1,\dotsc, q_m\}$ such that $f(p_i)=q_i$ for $i=1,\dotsc, k$.
  Define $Q'=\{q_1',\dotsc, q_m'\}$ by $q_i'=p_i$ for $i\le k$ and $q_i'=q_i$ for $i>k$,
  together with a bijection $g: Q'\to Q$ given by $g(q_i')=q_i$.
  Then $g$ introduces partial orders $<_{Q'}$ and $\evord_{Q'}$ on $Q'$ by
  $x <_{Q'} y$ iff $g(x) <_Q g(y)$ and $x \evord_{Q'} y$ iff $g(x) \evord_Q g(y)$.
  Setting $S_{Q'}=g(S_Q)$ and $T_{Q'}=g(T_Q)$ ensures that $g$ is an isomorphism and $T_P=S_{Q'}=P\cap Q'$.
\end{proof}

A \emph{pomset} is an ipomset $P$ without interfaces, \ie with $S_P=T_P=\emptyset$.

We introduce discrete ipomsets and some subclasses of these which will be important in what follows:
conclists will be the objects in the \emph{category} of ipomsets to be defined below,
and starters and terminators will be used in ipomset decompositions.

\begin{definition}
  An ipomset $(U, {<}, {\evord}, S, T, \lambda)$ is
  \begin{itemize}
  \item \emph{discrete} if ${<}$ is empty (hence $\evord$ is total);
  \item a \emph{conclist} (short for ``concurrency list'') if it is a discrete pomset ($S=T=\emptyset$);
    \hfill --  $\square$
  \item a \emph{starter} if it is discrete and $T=U$;
    \hfill -- $\St$
  \item a \emph{terminator} if it is discrete and $S=U$; and
    \hfill -- $\Te$
  \item an \emph{identity} if it is both a starter and a terminator.
    \hfill -- $\Id$
  \end{itemize}
  As already indicated above, we denote
  by $\square$ the set of conclists,
  by $\St$ the set of starters, and
  by $\Te$ the set of terminators.
  We write $\Id=\St\cap \Te$,
  $\Omega=\St\cup \Te$, $\St_+=\St\setminus \Id$, and $\Te_+=\Te\setminus \Id$.
\end{definition}

A conclist is, thus, a pomset of the form $(U, \emptyset, \evord, \emptyset, \emptyset, \lambda)$.
%
The \emph{source interface} of an ipomset $(P, {<}, {\evord}, S, T, \lambda)$ as above is the conclist $(S,  {\evord\rest{S\times S}}, \lambda\rest{S})$
where ``${}\rest{}$'' denotes restriction;
the \emph{target interface} of $P$ is the conclist $(T, {\evord\rest{T\times T}}, \lambda\rest{T})$.

A starter $\ilo{S}{U}{U}$ is an ipomset which starts the events in $U\setminus S$ and lets all events continue.
Similarly, a terminator $\ilo{U}{U}{T}$ is an ipomset in which all events are already active at the beginning
and the events in $U\setminus T$ are terminated.

We call a starter or terminator \emph{elementary} if $|S| = |P| - 1$, resp.\ $|T| = |P| - 1$, that is, if
it starts or terminates exactly one event.
In the following, a discrete ipomset will be represented by a vertical ordering of its elements following the event order, with elements in the source (resp.\ target) set preceded (resp.\ followed) by the symbol $\ibullet$.
For example, the discrete ipomset 
\[
  (\{x_1, x_2, x_3\}, \emptyset, \{(x_i,x_j)\mid i<j\}, \{x_1, x_2\}, \{x_2, x_3\},\{(x_1, a), (x_2, b), (x_3, c)\})
\]
is represented by $\loset{\ibullet a\pibullet \\ \ibullet b \ibullet \\ \pibullet c \ibullet}$.

The \emph{width} $\wid(P)$ of an ipomset $P$ is the cardinality of a maximal $<$-antichain.
(So, for example, the above ipomset has width $3$.)

We recall the definition of the gluing of ipomsets,
an operation that extends concatenation of words and serial composition of pomsets \cite{DBLP:journals/tcs/Gischer88}.
The intuition is that in a gluing $P*Q$,
the events of $P$ precede those of $Q$, \emph{except} for events which are in the target interface of $P$.
These events are continued in $Q$, across the gluing;
hence we require the target interface of $P$ to be isomorphic to the source interface of $Q$.
The underlying set of $P*Q$ is then given as the union of the two, but counting the continuing events only once.

\begin{definition} 
  \label{de:gluing}
  Let $P$ and $Q$ be two concrete ipomsets with $T_P\simeq S_Q$.
  The \emph{gluing} of $P$ and $Q$ is defined
  as $P*Q=(R, {<}, {\evord}, S, T, \lambda)$ where:
  \begin{enumerate}
  \item $R=(P\sqcup Q)_{x=f(x)}$, the quotient of the disjoint union under the unique isomorphism $f: T_P\to S_Q$;
  \item ${<}=
    \{(i(x), i(y))\mid x<_P y\}
    \cup \{(j(x), j(y))\mid x<_Q y\}
    \cup \{(i(x), j(y))\mid x\in P\setminus T_P, y\in Q\setminus S_Q\}$,
    where $i: P\to R$ and $j: Q\to R$ are the canonical injections;
  \item ${\evord}=(\{(i(x), i(y)\mid x\evord_P y\}\cup \{(j(x), j(y)\mid x\evord_Q y\})$;
  \item $S=i(S_P)$; $T=j(T_Q)$; and
  \item $\lambda(i(x))=\lambda_P(x)$, $\lambda(j(x))=\lambda_Q(x)$.
  \end{enumerate}
\end{definition}

\begin{figure}[tbp]
	\centering
	\begin{tikzpicture}[x=.35cm, y=.5cm]
			\begin{scope}
					\node (1) at (2,2) {$\vphantom{bd}a$};
					\node (2) at (0,0) {$b$};
					\node (3) at (4,0) {$\vphantom{bd}c\ibullet$};
					\path (2) edge (3);
					\path (1) edge[densely dashed, gray] (2);
					\path (3) edge[densely dashed, gray] (1);
					\node (ast) at (5.5,1) {$\ast$};
					\node (4) at (7,2) {$d$};
					\node (5) at (7,0) {$\vphantom{bd}\ibullet c$};
					\path (4) edge[densely dashed, gray] (5);
					\node (equals) at (8.5,1) {$=$};
					\node (6) at (10,2) {$\vphantom{bd}a$};
					\node (7) at (10,0) {$b$};
					\node (8) at (14,2) {$d$};
					\node (9) at (14,0) {$\vphantom{bd}c$};
					\path (6) edge (8);
					\path (7) edge (9);
					\path (7) edge (8);
					\path (8) edge[densely dashed, gray] (9);
					\path (9) edge[densely dashed, gray] (6);
					\path (6) edge[densely dashed, gray] (7);
				\end{scope}
		\end{tikzpicture}
	\caption{Gluing composition of ipomsets.}
	\label{fi:compositions}
\end{figure}
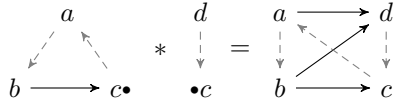

Figure \ref{fi:compositions} shows an example.
The relation $<$ is automatically transitive.
Note that composition is not cancellative:
for example, $a\ibullet * \loset{\ibullet a\\\pibullet  a} = a\ibullet * \loset{\pibullet  a\\\ibullet a} = \loset{a\\a}$.

Gluings of isomorphic ipomsets are isomorphic.
On isomorphism classes,
gluing is associative, and
ipomsets in $\Id$ are identities for $*$.
The next lemma extends Lemma~\ref{le:isoeq} and follows directly from it;
it shows that when gluing ipomsets, we may choose representatives such that isomorphisms become equalities.

\begin{lemma}
  \label{le:isoeqglu}
  For any concrete ipomsets $P$ and $Q$ with $T_P\simeq S_Q$, there exist $Q'\simeq Q$ and $R\simeq P*Q$
  such that
  $T_P=S_{Q'}=P\cap Q'$,
  $R=P\cup Q'$,
  ${<_R}={<_P}\cup {<_{Q'}}\cup (P\setminus Q')\times (Q'\setminus P)$,
  ${\evord_R}={\evord_P}\cup {\evord_{Q'}}$,
  $S_R=S_P$, and
  $T_R=T_{Q'}$. 
\end{lemma}

Ipomsets may be \emph{refined} by removing concurrency and expanding precedence.
The inverse to refinement is called \emph{subsumption}. Formally:

\begin{definition}
  \label{de:subsu}
  A \emph{subsumption} of an ipomset $P$ by $Q$
  is a bijection $f: P\to Q$ for which
  \begin{enumerate}
  \item $f(S_P)=S_Q$; $f(T_P)=T_Q$; $\lambda_Q\circ f=\lambda_P$;
  \item $f(x)<_Q f(y)$ implies $x<_P y$; and
  \item \label{en:subsu.evord}
    $x\evord_P y$ implies $f(x)\evord_Q f(y)$.
  \end{enumerate}
\end{definition}

We write $P\subseq Q$ if there is a subsumption $f: P\to Q$ and $P\subs Q$ if $P\subseq Q$ and
$P\not\simeq Q$.
Intuitively, $P$ has more order and less concurrency than $Q$.
Thus, subsumptions preserve interfaces and labels but may remove precedence order and add event order.
Isomorphisms of ipomsets are precisely invertible subsumptions;
but note that contrary to isomorphisms, subsumptions may not be unique.
The following lemma is a trivial consequence of the definitions.

\begin{lemma}
  \label{le:subsglue}
  Let $P$, $Q$, $P'$, $Q'$ be ipomsets such that
  $T_P=S_Q$, $T_{P'}=S_{Q'}$, $P\subseq Q$, and $P'\subseq Q'$.
  Then $P*P'\subseq Q*Q'$.
\end{lemma}

\begin{example}
  \label{ex:subsu}
  \mbox{}
  \begin{itemize}
  \item In Figure~\ref{fi:iposets1} there is a sequence of proper subsumptions from left to right:
    \begin{equation*}
      \ibullet a c b\subs
      \loset{\ibullet a\\ \pibullet c}*b\subs
      \loset{\ibullet a b\\ c}
    \end{equation*}
  \item The word $a b$ is subsumed by two different pomsets: $a b\subs \loset{a\\b}$ and $a b\subs \loset{b\\a}$.
  \item The fact that $a a\subs \loset{a\\a}$ is witnessed by two different bijections
    $f_1, f_2: a a\to \loset{a\\a}$:
    $f_1$ maps the $<$-minimal $a$ to the $\evord$-minimal $a$,
    and $f_2$ maps it to the $\evord$-maximal $a$ instead.
  \end{itemize}
\end{example}

\begin{definition}
  A \emph{step decomposition} of an ipomset $P$
  is a presentation
  \[
    P=P_1*\dotsc* P_n
  \]
  as a gluing of starters and terminators.
  The step decomposition is \emph{dense} if all of $P_1,\dotsc, P_n$ are elementary;
  it is \emph{sparse} if it is an alternating sequence of proper starters and terminators.
\end{definition}

An ipomset $P$ admits a step decomposition
if and only if $<_P$ 
is an \emph{interval order}~\cite{book/Fishburn85},
\ie if it admits an interval representation
given by functions $b, e: (P, {<_P})\to (\Real, {<_\Real})$ such that
\begin{itemize}
\item $b(x)\le_\Real e(x)$ for all $x\in P$ and
\item $x<_P y$ iff $e(x)<_\Real b(y)$ for all $x, y\in P$.
\end{itemize}
That is, every element $x$ of $P$ is associated with a real interval $[b(x), e(x)]$
such that $x<y$ in $P$ iff the interval of $x$ ends before the one of $y$ begins.
The ipomset of Figure~\ref{fi:n} is interval.
We will only treat interval ipomsets in this paper and thus omit the qualification ``interval''.


The set of (interval) ipomsets is written $\iPoms$.

\begin{lemma}[{%
    \cite[Proposition 3.5]{DBLP:journals/fuin/FahrenbergZ24};
    \cite[Lemma 4]{AMRANE2025115156}}]
  \label{le:ipomsparse}
  Let $P$ be an (interval) ipomset.
  \begin{itemize}
  \item $P$ has a unique sparse step decomposition.
  \item Every dense decomposition $P=P_1*\dotsm*P_n$ has the same length $n$.
  \end{itemize}
\end{lemma}

Showing existence of sparse decompositions is easy
and consists of gluing starters and terminators until no more such gluing is possible; showing uniqueness is more tedious.

We introduce special notations for starters and terminators to more clearly specify the conclists of events which are started or terminated.
For a conclist $U$ and subsets $A,B\subseteq U$ we write 
\begin{itemize}
\item
$\starter{U}{A}=\ilo{U\setminus A}{U}{U}=(U, {\evord}, U\setminus A, U, \lambda)$
and
\item
$\terminator{U}{B}=\ilo{U}{U}{U\setminus B}=(U, {\evord}, U, U\setminus B, \lambda)$.
\end{itemize}
The intuition is that the starter $\starter{U}{A}$ does nothing but start the events in $A=U\setminus S_U$
and the terminator $\terminator{U}{B}$ terminates the events in $B=U\setminus T_U$.

\section{A Categorical View of $\iiPoms$}
\label{se:catiipoms}
The following is clear, see also \cite[Proposition~1]{DBLP:journals/iandc/FahrenbergJSZ22}.

\begin{proposition}
  \label{prop:ipomsetCat}
  Ipomsets form a category $\iPoms$:
  \begin{itemize}
  \item objects are conclists;
  \item morphisms in $\iPoms(U, V)$ are ipomsets $P$ with $S_P=U$ and $T_P=V$;
  \item composition of morphisms is gluing;
  \item identities are $\id_U=\ilo{U}{U}{U}\in \iPoms(U, U)$. 
  \end{itemize}
\end{proposition}

\begin{proof}
  The only statement needing proof is that composition is associative with units $\ilo{U}{U}{U}$,
  and these properties are easily verified for gluing composition on (isomorphism classes of) ipomsets.
\end{proof}

\subsection{Ipomsets with transitive event order}
\label{se:ipoms-trans}


Our definition of ipomsets is different from the ones used in
\cite{DBLP:journals/mscs/FahrenbergJSZ21, DBLP:journals/lmcs/FahrenbergJSZ24},
which require event order to be transitive.
Here we make precise the relation between the two definitions.

A \emph{concrete ipomset with transitive event order} is a structure
$(P, {<}, {\evord}, S, T, \lambda)$ as in Definition \ref{de:ipoms},
except that $\evord$ is required to be a strict partial order,
and the requirement on the union of $<$ and $\evord$ is reduced to demanding that
\emph{at least one} of the following holds for every $x,y\in P$:
\[x<y,\qquad y<x,\qquad x\evord y,\qquad y\evord x,\qquad x=y.\]

Any ipomset as of Definition \ref{de:ipoms} may be turned into one with transitive event order
by transitively closing $\evord$, \ie
mapping $(P, {<}, {\evord}, S, T, \lambda)$ to
\begin{equation*}
  F((P, {<}, {\evord}, S, T, \lambda)) = (P, {<}, {\evord^+}, S, T, \lambda).
\end{equation*}
Conversely, any ipomset with transitive event order may be turned into the other type
by removing the \emph{inessential} part of $\evord$,
\ie mapping $(P, {<}, {\evord}, S, T, \lambda)$ to
\begin{equation*}
  G((P, {<}, {\evord}, S, T, \lambda)) = (P, {<}, {\evord'}, S, T, \lambda)
\end{equation*}
with ${\evord'}={\evord}\cap {\not<}\cap {\not>}$.

The definitions of isomorphism and subsumption of ipomsets with transitive event order
have to take into account that some of the event order may be inessential in the sense above,
changing item \ref{en:subsu.evord} of Definition \ref{de:subsu} to demand that
\begin{itemize}
\item if $x\evord_P y$, $x\not< y$, and $y\not< y$, then $f(x)\evord_Q f(y)$.
\end{itemize}
Once this is in place, the mappings $F$ and $G$ above extend to an isomorphism of categories
between $\iiPoms$ and the category of isomorphism classes of ipomsets with transitive event order.
The two notions are, thus, equivalent.

\subsection{Step sequences}
\label{sse:step-seq}
\label{sse:sandt}

We develop a representation of the category $\iPoms$ by generators and relations,
using the step decompositions introduced in Section \ref{se:ipomsets-step}.
We view step decompositions as sequences of starters and terminators,
that is, as words over the (infinite) alphabet $\Omega=\St\cup \Te$.




Let $\bOmega$ be the directed multigraph given as follows:
\begin{itemize}
\item Vertices are conclists.
\item Edges in $\bOmega(U, V)$ are starters and terminators $P$ with $S_P=U$ and $T_P=V$.
\end{itemize}
Note that
$\bOmega(U, V)\subseteq \St$ or $\bOmega(U, V)\subseteq \Te$
for all $U, V\in \square$.

Let $\Coh$
be the category freely generated by $\bOmega$.
Non-identity morphisms in $\Coh(U,V)$ are words
$P_1\dotsc P_n\in \Omega^+$,
\ie~such that $T_{P_i}=S_{P_{i+1}}$ for all $i=1,\dots, n-1$.
Such words are called \emph{coherent} in \cite{AMRANE2025115156}.
Note that $P_1\dots P_n$ is coherent iff the gluing $P_1*\dots* P_n$ is defined.

%

Let $\sim$ be the congruence on $\Coh$ generated by the relations
\begin{equation}
  \label{eq:sim-stepseq}
  \begin{alignedat}{2}
    P Q &\sim P*Q \qquad &&(P, Q\in \St \text{ or } P, Q\in \Te),
    \\
    \id_U &\sim \ilo{U}{U}{U} \qquad &&(U\in \square). 
  \end{alignedat}
\end{equation}
The first of these allows to compose subsequent starters and subsequent terminators, and the second identifies the (freely generated) identities at $U$
with the corresponding ipomset identities in $\Id$.
(Note that the gluing of two starters is again a starter, and similarly for terminators; but ``mixed'' gluings do not have this property.)
We let $\Cohsim$ denote the quotient of $\Coh$ under $\sim$.

\begin{figure}
  \centering
  \begin{tikzcd}[sep=large]
    \bOmega\ar[r,hook] &\Coh\ar[d,"{[-]_\sim}" swap]\ar[r,"\Psip",shift left=1] & \iPoms\ar[d, equal]\ar[l,"\Phip",shift left=1] \\
    & \Cohsim\ar[r,"\Psi",shift left=1] & \iPoms\ar[l,"\Phi",shift left=1]
  \end{tikzcd}
  \caption{%
    Relationship between step sequence and ipomset categories.
    The arrows $\Psip$, $\Phi$, $\Psi$ and ${[-]_\sim}$ denote functors; $\Phip$ is not a functor.
    Note that all the maps in the diagram are identities on objects.}
  \label{fi:ipomscohphipsi}
\end{figure}

Let $\Psip:\Coh\to \iPoms$ be the functor induced by the inclusion $\bOmega\hookrightarrow \iPoms$:
\begin{equation*}
  \Psip(U)=U,\qquad \Psip(P)=P.
\end{equation*}
Then $\Psip(P_1\dots P_n)=P_1*\dots*P_n$.
The following is straightforward.

\begin{lemma}
  If $P_1\dots P_n\sim Q_1\dots Q_m$, then $\Psip(P_1\dotsc P_n)=\Psip(Q_1\dots Q_m)$. 
\end{lemma}

Thus $\Psip$ induces a functor $\Psi:\Cohsim\to \iPoms$; we show below that $\Psi$ is an isomorphism of categories.
See Figure~\ref{fi:ipomscohphipsi} for an overview of the introduced structures and mappings.

A \emph{step sequence} \cite{conf/apn/AmraneBCF24} is a morphism in $\Cohsim$,
that is, an equivalence class of coherent words under $\sim$.
We redevelop the facts about step decompositions from Section \ref{se:ipomsets-step} in terms of step sequences.

\begin{lemma}
  \label{le:stepdecomp}
  For every $P\in \iPoms$ there exists $w\in \Cohsim$ such that $\Psi(w)=P$.
\end{lemma}

\begin{proof}
  This is the same as saying that every ipomset has a step decomposition.
\end{proof}

A word $P_1\dots P_n\in \Coh$ is \emph{dense} if all its elements are elementary, \ie start or terminate precisely one event.
It is \emph{sparse} if proper starters and terminators are alternating,
that is, for all $i=1,\dots, n-1$, $(P_i,P_{i+1}) \in (\St_+ \times \Te_+) \cup (\Te_+ \times \St_+)$.
By convention, identities $\id_U\in \Coh$ are both dense and sparse.
We let $\DCoh, \SCoh\subseteq \Coh$ denote the subsets of dense, resp.\ sparse coherent words.

\begin{lemma}
  \label{le:sparseuniq}
  Every step sequence contains exactly one sparse coherent word.
\end{lemma}

\begin{proof}
  Let $P_1\dotsm P_n\in\Coh$
  be a representative of a step sequence having minimal length.
  If $P_i$ and $P_{i+1}$ are both starters or both terminators, 
  then $P_1\dotsm P_{i-1}(P_iP_{i+1})P_{i+2}\dotsm P_n$
  is a shorter representative: a contradiction.
  Thus, $P_1\dotsm P_n\in\Coh$ is sparse.

  If $Q_1\dotsm Q_m\sim P_1\dotsm P_n$ is another sparse representative,
  then
  \[
    P_1*\dotsm*P_n=P=Q_1*\dotsm*Q_m
  \]
  are both sparse decompositions of $P$,
  and by Lemma \ref{le:ipomsparse} they are equal.
\end{proof}

\begin{example}
  The unique sparse step sequence corresponding to the ipomset in Figure~\ref{fi:n} is
  \begin{equation*}
    \bigloset{ \pibullet  a \ibullet \\ \ibullet c \ibullet} 
    \bigloset{\ibullet a \ibullet \\ \ibullet c\pibullet }
    \bigloset{\ibullet a \ibullet \\ \pibullet a \ibullet}
    \bigloset{\ibullet a\pibullet  \\ \ibullet a \ibullet}
    \bigloset{\pibullet  b \ibullet \\ \ibullet a \ibullet}
    \bigloset{ \ibullet b \\ \ibullet a}:
  \end{equation*}
  it first starts the first $a$, then terminates $c$,
  then starts the second $a$,
  terminates the first $a$,
  then starts $b$ and finally terminates both $b$ and the second $a$.
  The corresponding dense step sequences are
  \begin{align*}
    & \bigloset{\pibullet  a \ibullet \\ \ibullet c \ibullet} 
    \bigloset{\ibullet a \ibullet \\ \ibullet c\pibullet }
    \bigloset{\ibullet a \ibullet \\ \pibullet a \ibullet}
    \bigloset{\ibullet a\pibullet  \\ \ibullet a \ibullet}
    \bigloset{\pibullet  b \ibullet \\ \ibullet a \ibullet}
    \bigloset{ \ibullet b\pibullet  \\ \ibullet a \ibullet}
    \bigloset{\ibullet a} \\
    {}={}
    & \bigloset{\pibullet  a \ibullet \\ \ibullet c \ibullet} 
    \bigloset{\ibullet a \ibullet \\ \ibullet c\pibullet }
    \bigloset{\ibullet a \ibullet \\ \pibullet a \ibullet}
    \bigloset{\ibullet a\pibullet  \\ \ibullet a \ibullet}
    \bigloset{\pibullet  b \ibullet \\ \ibullet a \ibullet}
    \bigloset{ \ibullet b \ibullet \\ \ibullet a\pibullet }
    \bigloset{\ibullet b},
  \end{align*}
  which differ only in the order in which $b$ and $a$ are terminated at the end.
\end{example}

Using Lemmas \ref{le:stepdecomp} and \ref{le:sparseuniq},
we may now define a functor $\Phi: \iPoms \to \Cohsim$ which will serve as inverse to $\Psi$.
First, for $P\in \iPoms$ let $\Phip(P)\in \Coh$ be its unique sparse step decomposition;
this defines a mapping $\Phip: \iPoms\to \Coh$.
Now define $\Phi$ by $\Phi(P)=[\Phip(w)]_\sim$.

\begin{theorem}
  \label{th:genipoms}
  $\Phi$ is a functor, $\Psi \circ \Phi = \Id_{\iPoms}$, and $\Phi \circ \Psi = \Id_{\Cohsim}$.
  Hence $\Phi: \iPoms\leftrightarrows \Cohsim: \Psi$ is an isomorphism of categories.
\end{theorem}

\begin{proof}
  We have $\Phi(\id_U)=[\Phip(\id_U)]_\sim=[\Id_U]_\sim=\id_U$ using the relations (\ref{eq:sim-stepseq}).
  Let $P, Q\in\iPoms\setminus\Id$ be ipomsets such $T_P=S_Q$ 
  and $P=P_1*\dotsm*P_n$ and $Q=Q_1*\dotsm*Q_m$ the unique sparse decompositions.
  If $P_n$ is a starter and $Q_1$ is a terminator or vice versa,
  then $P*Q=P_1*\dotsm*P_n*Q_1*\dotsm*Q_m$ is sparse, so that
  \[
    \Phip(P*Q)=P_1\dotsm P_nQ_1\dotsm Q_m
    =
    \Phip(P)\Phip(Q).
  \]
  If $P_n$ and $Q_1$ are both starters or both terminators, then
  \[
    P*Q=P_1*\dotsm*P_{n-1}*(P_n*Q_1)*Q_2*\dotsm*Q_m
  \]
  is a sparse decomposition, hence
  \[
    \Phip(P*Q)=
    P_1\dotsm P_{n-1}(P_n*Q_1)Q_2\dotsm Q_m
    \sim
    P_1\dotsm P_{n-1}P_nQ_1Q_2\dotsm Q_m
    =
    \Phip(P)\Phip(Q).
  \]
  In both cases, $\Phi(P*Q)=[\Phip(P*Q)]_\sim=[\Phip(P)\Phip(Q)]_\sim=[\Phip(P)]_\sim[\Phip(Q)]_\sim=\Phi(P)\Phi(Q)$.	
  The case when $P$ or $Q$ is an identity can be handled in a similar way.
  As a consequence, $\Phi$ is a functor.
  
  The composition $\Psip\Phip$ is clearly the identity on $\iPoms$
  so $\Psi\Phi$ is also an identity.
  For $P_1\dotsm P_n\in \Coh$, 
  $\Phip\Psip(P_1\dotsm P_n)$ is the unique sparse representative
  of $P_1\dotsm P_n$.
  In particular, $P_1\dotsm P_n\sim \Phip\Psip(P_1\dotsm P_n)$ and then $\Phi\Psi=\Id_{\Coh_\sim}$.
\end{proof}

Given that $\Cohsim$ and $\iPoms$ are isomorphic,
we will often confuse the two
and, for example, write $[w]_\sim$ instead of $\Psip(w)=\Psi([w]_\sim)$ for $w\in\Coh$.

\begin{corollary}
  The category $\iPoms$ is generated by the directed multigraph $\bOmega$ using gluing composition,
  under the identities \eqref{eq:sim-stepseq}.
\end{corollary}

\section{Subsumptions in Step Sequences}
\label{se:step-subsumptions}

In this section,
we extend the equivalence between ipomsets and step sequences from Theorem~\ref{th:genipoms}
to also cover subsumptions.

Define a partial preorder on $\DCoh$ generated by
\begin{alignat}{2}
  \starter{(U\setminus b)}{a}\cdot \starter {U} {b}
  &\preceq
  \starter{(U\setminus a)}{b}\cdot \starter {U} {a}
  \qquad &&(a\ne b)
  \label{eq:subselm.stst} \\
  \terminator {U} {b}\cdot \terminator{(U\setminus b)}{a}
  &\preceq
  \terminator {U} {a}\cdot \terminator{(U\setminus a)}{b}
  \qquad &&(a\ne b)
  \label{eq:subselm.tete} \\
  \terminator {(U \setminus a)} b \cdot \starter {(U \setminus b)} a
  &\preceq
  \starter U a \cdot \terminator U b
  \qquad &&(a\ne b)
  \label{eq:subselm.stte} \\
  w\preceq w' &\implies v w y\preceq v w' y
  \notag
\end{alignat}
These relations swap elements of coherent words,
taking care of adjusting them to preserve coherency.
Transpositions of type \eqref{eq:subselm.stst} and \eqref{eq:subselm.tete}
swap subsequent starters, resp.\ subsequent terminators;
these are, in fact, equivalences.
Transpositions of type \eqref{eq:subselm.stte} swap a starter with a terminator,
introducing a proper subsumption.

\begin{lemma}
  \label{le:subselm-sound}
  For any $w, w'\in \DCoh$, $w\preceq w'$ implies $[w]_\sim\subseq [w']_\sim$.
\end{lemma}

\begin{proof}
  The three elementary cases above all define subsumptions,
  and by Lemma \ref{le:subsglue} gluing preserves subsumptions.
\end{proof}

\begin{remark}
  In the context of HDAs,
  \cite{DBLP:journals/tcs/Glabbeek06}
  defines a notion of \emph{adjacency} for paths
  which consists of precisely the analogues of the transformations that we define above.
  Adjacency is then used to define \emph{homotopy} of paths,
  which is the equivalence closure of adjacency.
  Taking only the reflexive and transitive closure,
  we will instead use it to define subsumptions.
\end{remark}

\begin{figure}[tbp]
  \centering
  \begin{tikzpicture}[x=.95cm, scale=1, every node/.style={transform shape}]
    \def\possh{-1.3}

    \begin{scope}[shift={(0,0)}]
      \def\hw{0.3}
      \filldraw[fill=green!50!white,-](0,1.2)--(1.8,1.2)--(1.8,1.2+\hw)--(0,1.2+\hw) -- (0, 1.2);
      \filldraw[fill=red!50!white,-](-1.2,0.7)--(1.2,0.7)--(1.2,0.7+\hw)--(-1.2,0.7+\hw)--(-1.2,0.7);
      \filldraw[fill=blue!20!white,-](-.6,0.2)--(0.6,0.2)--(0.6,0.2+\hw)--(-.6,0.2+\hw)--(-.6,0.2);
      \node at (0.8,1.2+\hw*0.5) {$a$};
      \node at (0.2,0.7+\hw*0.5) {$b$};
      \node at (0.3,0.2+\hw*0.5) {$c$};
      \path (-0.6, 1.2+\hw) edge[-,dashed] (-0.6,0.2);
      \path (0, 1.2+\hw) edge[-,dashed] (0,0.2);
      \path (0.6, 1.2+\hw) edge[-,dashed] (0.6,0.2);
      \path (1.2, 1.2+\hw) edge[-,dashed] (1.2,0.2);
      \path (1.8, 1.2+\hw) edge[-,dashed] (1.8,0.2);
      \node (0) at (-1.2, -0.1) {$1$};
      \node (1) at (-0.6, -0.1) {$2$};
      \node (2) at (-0, -0.1) {$3$};
      \node (3) at (0.6, -0.1) {$4$};
      \node (4) at (1.2, -0.1) {$5$};
      \node (5) at (1.8, -0.1) {$6$};
      \node at (.3,-2) {$%
        b \ibullet
        \bigloset{\ibullet b\ibullet \\ \pibullet c\ibullet}
        \bigloset{\pibullet a \ibullet \\ \ibullet b \ibullet \\ \ibullet c \ibullet}
        \bigloset{\ibullet a  \ibullet \\ \ibullet b \ibullet \\ \ibullet c\pibullet }
        \bigloset{\ibullet a \ibullet \\ \ibullet b\pibullet }
        \ibullet a$};
    \end{scope}
    
    \begin{scope}[shift={(5.5,0)}]
      \def\hw{0.3}
      \filldraw[fill=green!50!white,-](0,1.2)--(1.8,1.2)--(1.8,1.2+\hw)--(0,1.2+\hw) -- (0, 1.2);
      \filldraw[fill=red!50!white,-](-1.2,0.7)--(.6,0.7)--(.6,0.7+\hw)--(-1.2,0.7+\hw)--(-1.2,0.7);
      \filldraw[fill=blue!20!white,-](-.6,0.2)--(1.2,0.2)--(1.2,0.2+\hw)--(-.6,0.2+\hw)--(-.6,0.2);
      \node at (0.8,1.2+\hw*0.5) {$a$};
      \node at (-0.2,0.7+\hw*0.5) {$b$};
      \node at (0.4,0.2+\hw*0.5) {$c$};
      \path (-0.6, 1.2+\hw) edge[-,dashed] (-0.6,0.2);
      \path (0, 1.2+\hw) edge[-,dashed] (0,0.2);
      \path (0.6, 1.2+\hw) edge[-,dashed] (0.6,0.2);
      \path (1.2, 1.2+\hw) edge[-,dashed] (1.2,0.2);
      \path (1.8, 1.2+\hw) edge[-,dashed] (1.8,0.2);
      \node (0) at (-1.2, -0.1) {$1$};
      \node (1) at (-0.6, -0.1) {$2$};
      \node (2) at (-0, -0.1) {$3$};
      \node (3) at (0.6, -0.1) {$4$};
      \node (4) at (1.2, -0.1) {$5$};
      \node (5) at (1.8, -0.1) {$6$};
      \node at (.3,-2) {$%
        b \ibullet
        \bigloset{\ibullet b\ibullet \\ \pibullet c\ibullet}
        \bigloset{\pibullet a \ibullet \\ \ibullet b \ibullet \\ \ibullet c \ibullet}
        \bigloset{\ibullet a  \ibullet \\ \ibullet b \pibullet \\ \ibullet c\ibullet }
        \bigloset{\ibullet a \ibullet \\ \ibullet c\pibullet }
        \ibullet a$};
    \end{scope}
    
    \begin{scope}[shift={(11,0)}]
      \def\hw{0.3}
      \filldraw[fill=green!50!white,-](0.6,1.2)--(1.8,1.2)--(1.8,1.2+\hw)--(0.6,1.2+\hw) -- (0.6,1.2);
      \filldraw[fill=red!50!white,-](-1.2,0.7)--(0,0.7)--(0,0.7+\hw)--(-1.2,0.7+\hw)--(-1.2,0.7);
      \filldraw[fill=blue!20!white,-](-.6,0.2)--(1.2,0.2)--(1.2,0.2+\hw)--(-.6,0.2+\hw)--(-.6,0.2);
      \node at (0.8,1.2+\hw*0.5) {$a$};
      \node at (-0.2,0.7+\hw*0.5) {$b$};
      \node at (0.4,0.2+\hw*0.5) {$c$};
      \path (-0.6, 1.2+\hw) edge[-,dashed] (-0.6,0.2);
      \path (0, 1.2+\hw) edge[-,dashed] (0,0.2);
      \path (0.6, 1.2+\hw) edge[-,dashed] (0.6,0.2);
      \path (1.2, 1.2+\hw) edge[-,dashed] (1.2,0.2);
      \path (1.8, 1.2+\hw) edge[-,dashed] (1.8,0.2);
      \node (0) at (-1.2, -0.1) {$1$};
      \node (1) at (-0.6, -0.1) {$2$};
      \node (2) at (-0, -0.1) {$3$};
      \node (3) at (0.6, -0.1) {$4$};
      \node (4) at (1.2, -0.1) {$5$};
      \node (5) at (1.8, -0.1) {$6$};
      \node at (.3,-2) {$%
        b \ibullet
        \bigloset{\ibullet b\ibullet \\ \pibullet c\ibullet}
        \bigloset{\ibullet b\pibullet \\ \ibullet c\ibullet}
        \bigloset{\pibullet a \ibullet \\ \ibullet c \ibullet}
        \bigloset{\ibullet a \ibullet \\ \ibullet c\pibullet }
        \ibullet a$};
    \end{scope}
  \end{tikzpicture}

  \caption{%
    Interval representations of several ipomsets
    with corresponding dense coherent words,
    \cf Example~\ref{ex:swap}.}
  \label{fi:intervalrep}
\end{figure}
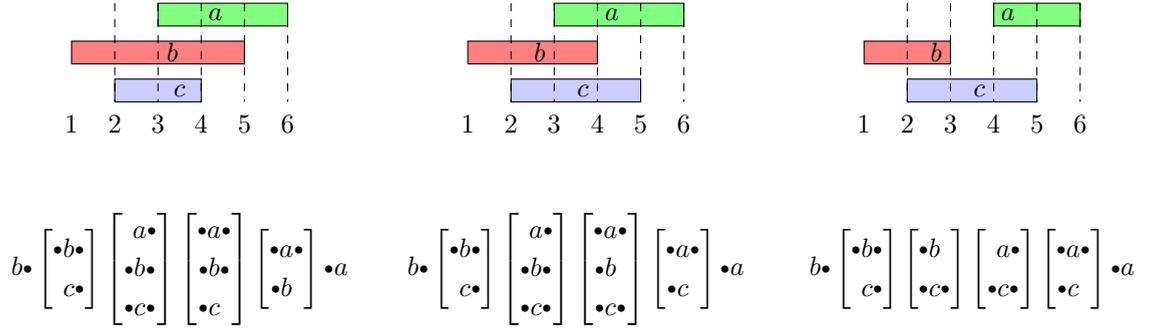

\begin{example}
  \label{ex:swap}
  Figure~\ref{fi:intervalrep} presents several interval representations of ipomsets,
  together with their corresponding dense coherent words.
  Progressing from left to right,
  the first transposition employed is of type \eqref{eq:subselm.tete}
  applied to the fourth and fifth elements,
  swapping termination of $c$ with termination of $b$.
  The second transposition is of type \eqref{eq:subselm.stte}
  and swaps the start of $a$ with the termination of $b$,
  creating a precedence $b<a$.
\end{example}

Our goal is now to show an inverse to Lemma \ref{le:subselm-sound},
showing that subsumptions of ipomsets are generated by the elementary transpositions
\eqref{eq:subselm.stst}, \eqref{eq:subselm.tete} and \eqref{eq:subselm.stte}.

For $w=P_1\dotsm P_n$ and $p\in P=[w]_\sim$ define
\begin{align*}
  \xst_w(p) &=
  \begin{cases}
    \min\{i\mid p\in P_i\} & \text{for $p\not\in S_P$},\\
    -\infty & \text{for $p\in S_P$},
  \end{cases} \\
  \xen_w(p) &=
  \begin{cases}
    \max\{i\mid p\in P_i\} & \text{for $p\not\in T_P$},\\
    +\infty & \text{for $p\in T_P$}.
  \end{cases}
\end{align*}
The next lemma shows that $(\xst_w,\xen_w)$ is a ``bijective'' interval representation of $P$.

\begin{lemma}
  We have
  $\xst_w(P)\cap \xen_w(P)=\emptyset$,
  and for every $i=1,\dotsc, n$ there is precisely one $p\in P$ with $\xst_w(p)=i$ or $\xen_w(p)=i$.
\end{lemma}

\begin{proof}
  Every $P_i$ is either a starter or a terminator, hence $\xst_w(P)\cap \xen_w(P)=\emptyset$.
  Further, every $P_i$ is elementary, implying the second claim.
\end{proof}

We may hence define a function $\phi_w:\{1,\dotsc,n\}\to P$
that tells which event starts or terminates at a given place,
given by $\phi_w(i)=p$ if $\xst_w(p)=i$ or $\xen(p)=i$.


\begin{lemma}
  \label{le:subselm-complete}
  If $P\subseq Q$, $P=[u]_\sim$, and $Q=[v]_\sim$, then $u\preceq v$.
\end{lemma}

\begin{proof}
  Let $u=P_1\dotsc P_n$ and $v=Q_1\dotsc Q_n$
  ($u$ and $v$ have the same length by Lemma \ref{le:ipomsparse}).
  Consider some cases:
  \begin{enumerate}
  \item
    $P_1$ is a starter.
    Let $p=\phi_u(1)$ be the event started at $P_1$, then $p\notin S_P=S_Q$.
    Let $m=\xst_v(p)$.
    Assume that there is $j\in \{1,\dotsc, m-1\}$ such that $Q_j$ is a terminator,
    then with $q=\phi_v(j)$ we would have $q<_Q p$ but $q\not<_P p$: a contradiction to $P\subseq Q$.
    
    Hence all of $Q_1\dotsc, Q_m$ are starters,
    so we can use transpositions of type \eqref{eq:subselm.stst}
    to move $Q_m$ to the very beginning of $v$.
    Now $P_1=Q_1$,
    and we can use induction on $P_2\dotsm P_n$ and $Q_2\dotsm Q_n$.


  \item
    $Q_1$ is a terminator.
    This uses a dual argument to the one above,
    showing that $u$ must begin with a sequence of terminators,
    one of which terminates $\phi_v(1)$,
    and then using transpositions of type \eqref{eq:subselm.tete}
    to move that terminator to the very beginning.
  \item
    $P_1$ is a terminator and $Q_1$ is a starter.
    Let $p=\phi_u(1)$ be the event terminated at $P_1$, then $p\in S_P=S_Q$ and $p\notin T_P=T_Q$.
    Let $k=\xen_v(p)$ be the index at which $p$ is terminated in $v$.

    None of $Q_1,\dotsc, Q_{k-1}$ start or terminate $p$,
    so we can use transpositions of type \eqref{eq:subselm.tete} or \eqref{eq:subselm.stte}
    (from right to left)
    to move $Q_k$ to the beginning of $v$.
    Now $P_1=Q_1$,
    and we can again use induction on $P_2\dotsm P_n$ and $Q_2\dotsm Q_n$.
  \end{enumerate}
\end{proof}

\begin{example}
  Let $P = \loset{a \\ b}$, $Q = ab$,
  and $b\ibullet \loset{\pibullet a \ibullet \\ \ibullet b \ibullet} \loset{\ibullet a \ibullet \\ \ibullet b \pibullet } \ibullet a$
  and $a\ibullet\, \ibullet a b\ibullet\, \ibullet b$
  be dense step sequences corresponding to $P$ resp.\ $Q$.
  An example of a sequence as in Lemma \ref{le:subselm-complete} that underlines the fact that $Q \subs P$ is
  \begin{equation*}
  \begin{array}{r@{\,}c@{\,}c@{\,}c@{\,}c}
    w_1 &= b\ibullet &\loset{\pibullet a \ibullet \\ \ibullet b \ibullet} &\loset{\ibullet a \ibullet \\ \ibullet b\pibullet } &\ibullet a, \\
    w_2 &= a\ibullet &\loset{\ibullet a \ibullet \\ \pibullet b \ibullet} &\loset{\ibullet a \ibullet \\ \ibullet b\pibullet } &\ibullet a, \\
    w_3 &=  a\ibullet &\loset{\ibullet a \ibullet \\ \pibullet b \ibullet} &\loset{\ibullet a \pibullet \\ \ibullet b\ibullet} &\ibullet b, \\
    w_4 &=  a\ibullet\, &\ibullet a& b\ibullet\, &\ibullet b.
  \end{array}
  \end{equation*}
\end{example}


\begin{theorem}
  \label{th:subsumption}
  For $P, Q\in \iPoms$, the following conditions are equivalent.
\begin{enumerate}
\item
	$P\subseq Q$.
\item
	$v\preceq w$ for all $v,w\in\DCoh$ such that $[v]_\sim = P$, $[w]_\sim=Q$.
\item
	$v\preceq w$ for some $v,w\in\DCoh$ such that $[v]_\sim = P$, $[w]_\sim=Q$.	
\end{enumerate}  
\end{theorem}

\begin{proof}
	Implication 1.\;$\Rightarrow$ 2.\@ is Lemma \ref{le:subselm-complete},
	2.\;$\Rightarrow$ 3.\@ follows by the existence of dense step decompositions,
	and 3.\;$\Rightarrow$ 1.\@ is Lemma \ref{le:subselm-sound}.
\end{proof}


\begin{corollary}
  Every subsumption $P\subseq Q$ is a composition of elementary 
  subsumptions of the form 
  $P'*\terminator {(U \setminus a)} b * \starter {(U \setminus b)} a *P''
  \subs
  P'*\starter U a * \terminator U b*P''$.
\end{corollary}

\begin{proof}
  This follows from Lemma \ref{le:subselm-complete} and the definition of $\preceq$.
  (Relations \eqref{eq:subselm.stst} and \eqref{eq:subselm.tete}
  can be skipped since they induce
  isomorphisms.)
\end{proof}


\section{Higher-Dimensional Automata and ST-Automata}
\label{se:hda-st-a}

We now transfer the isomorphism between ipomsets and step sequences to the operational side.
We recall higher-dimensional automata which generate ipomsets
and introduce ST-automata which generate step sequences,
and we clarify their relation.


\subsection{Higher-dimensional automata}
\label{sse:hda}

We give a quick introduction to higher-dimensional automata and their languages
and refer the interested reader to \cite{DBLP:journals/fuin/FahrenbergZ24, AMRANE2025115156} for details and examples.

A \emph{precubical set}
\begin{equation*}
  X=(X, {\ev}, \{\delta_{A, U}^0, \delta_{A, U}^1\mid U\in \square, A\subseteq U\})
\end{equation*}
consists of a set of \emph{cells} $X$
together with a function $\ev: X\to \square$ which to every cell assigns a conclist of concurrent events which are active in it.
We write $X[U]=\{q\in X\mid \ev(q)=U\}$ for the cells of type $U$.
For every $U\in \square$ and $A\subseteq U$ there are \emph{face maps}
$\delta_{A}^0, \delta_{A}^1: X[U]\to X[U\setminus A]$
(we often omit the extra subscript $U$)
which satisfy
\begin{equation}
  \label{eq:prid}
  \text{%
    $\delta_A^\nu \delta_B^\mu = \delta_B^\mu \delta_A^\nu$
    for $A\cap B=\emptyset$ and $\nu, \mu\in\{0, 1\}$.%
  }
\end{equation}
The \emph{upper} face maps $\delta_A^1$ terminate events in $A$
and the \emph{lower} face maps $\delta_A^0$ transform a cell $q$
into one in which the events in $A$ have not yet started.

A \emph{higher-dimensional automaton} (\emph{HDA}) $\hda{H}=(\cell{H}, \bot_\cell{H}, \top_\cell{H})$
is a precubical set together with subsets $\bot_\cell{H}, \top_\cell{H}\subseteq \cell{H}$
of \emph{start} and \emph{accept} cells.
We do not generally assume HDAs to be finite,
but will do so in Section \ref{sec:MSO}.
The \emph{dimension} of an HDA $\hda{H}$ is $\dim(\hda{H})=\sup\{|\ev(q)|\mid q\in \cell{H}\}\in \Nat\cup \{\infty\}$.

\begin{example}
A standard automaton is the same as a one-dimensional HDA $\hda{H}$
with the property that for all $q \in \bot_\cell{H} \cup \top_\cell{H}$, $\ev(q) = \emptyset$:
cells in $\cell{H}[\emptyset]$ are states,
cells in $\cell{H}[\{a\}]$ for $a\in \Sigma$ are $a$-labelled transitions,
and face maps $\delta_{\{a\}}^0$ and $\delta_{\{a\}}^1$
attach source and target states to transitions.
In contrast to ordinary automata we allow start and accept \emph{transitions}
instead of merely states,
so languages of one-dimensional HDAs may contain words with interfaces.
\end{example}

\begin{example}
  \label{ex:hda}
	Figure~\ref{fi:abcube} shows a two-dimensional HDA as a combinatorial object (left)
	and in a geometric realisation (right).
	It consists of
	21 cells:
	states $\cell{H}_0 = \{v_1,\dots, v_8\}$ in which no event is active ($\ev(v_i) = \emptyset$),
	transitions $\cell{H}_1 = \{t_1,\dots, t_{10}\}$ in which one event is active (\eg $\ev(t_3) = \ev(t_4) = c$),
	squares $\cell{H}_2 = \{q_1, q_2, q_3\}$ where $\ev(q_1) = \loset{a\\c}$ and $\ev(q_2) = \ev(q_3) = \loset{a\\d}$.
	The arrows between cells in the left representation correspond to the face maps connecting them.
	For example, the upper face map $\delta^1_{a c}$ maps $q_1$ to $v_4$
	because the latter is the cell in which the active events $a$ and $c$ of $q_1$ have been terminated.
	On the right, face maps are used to glue cells,
	so that for example $\delta^1_{a c}(q_1)$ is glued to the top right of $q_1$.
	In this and other geometric realisations,
	when we have two concurrent events $a$ and $c$ with $a\evord c$, we will draw $a$ horizontally and $c$ vertically.
\end{example}

\begin{figure}[tbp]
	\centering
	\begin{tikzpicture}[x=.7cm, y=.62cm, every node/.style={transform shape}]
		\begin{scope}[y=.7cm, scale=.9]
			\node[circle,draw=black,fill=blue!20,inner sep=0pt,minimum size=15pt]
			(aa) at (0,0) {$\vphantom{hy}v_1$};
			\node[circle,draw=black,fill=blue!20,inner sep=0pt,minimum size=15pt]
			(ac) at (0,4) {$\vphantom{hy}v_2$};
			\node[circle,draw=black,fill=blue!20,inner sep=0pt,minimum size=15pt]
			(ca) at (4,0) {$\vphantom{hy}v_3$};
			\node[circle,draw=black,fill=blue!20,inner sep=0pt,minimum size=15pt]
			(cc) at (4,4) {$\vphantom{hy}v_4$};
			\node[circle,draw=black,fill=blue!20,inner sep=0pt,minimum size=15pt]
			(ae) at (0,8) {$\vphantom{hy}v_5$};
			\node[circle,draw=black,fill=blue!20,inner sep=0pt,minimum size=15pt]
			(ec) at (8,4) {$\vphantom{hy}v_6$};
			\node[circle,draw=black,fill=blue!20,inner sep=0pt,minimum size=15pt]
			(ce) at (4,8) {$\vphantom{hy}v_7$};
			\node[circle,draw=black,fill=blue!20,inner sep=0pt,minimum size=15pt]
			(ee) at (8,8) {$\vphantom{hy}v_8$};
			\node[circle,draw=black,fill=green!30,inner sep=0pt,minimum size=15pt]
			(ba) at (2,0) {$\vphantom{hy}t_1$};
			\node[circle,draw=black,fill=green!30,inner sep=0pt,minimum size=15pt]
			(bc) at (2,4) {$\vphantom{hy}t_2$};
			\node[circle,draw=black,fill=green!30,inner sep=0pt,minimum size=15pt]
			(ab) at (0,2) {$\vphantom{hy}t_3$};
			\node[circle,draw=black,fill=green!30,inner sep=0pt,minimum size=15pt]
			(ad) at (0,6) {$\vphantom{hy}t_5$};
			\node[circle,draw=black,fill=green!30,inner sep=0pt,minimum size=15pt]
			(be) at (2,8) {$\vphantom{hy}t_6$};
			\node[circle,draw=black,fill=green!30,inner sep=0pt,minimum size=15pt]
			(cd) at (4,6) {$\vphantom{hy}t_8$};
			\node[circle,draw=black,fill=green!30,inner sep=0pt,minimum size=15pt]
			(de) at (6,8) {$\vphantom{hy}t_9$};
			\node[circle,draw=black,fill=green!30,inner sep=0pt,minimum size=15pt]
			(dc) at (6,4) {$\vphantom{hy}t_7$};
			\node[circle,draw=black,fill=green!30,inner sep=0pt,minimum size=15pt]
			(ed) at (8,6) {$\vphantom{hy}t_{10}$};
			\node[circle,draw=black,fill=green!30,inner sep=0pt,minimum size=15pt]
			(cb) at (4,2) {$\vphantom{hy}t_{4}$};
			\node[circle,draw=black,fill=black!20,inner sep=0pt,minimum size=15pt]
			(bb) at (2,2) {$\vphantom{hy}q_1$};
			\node[circle,draw=black,fill=black!20,inner sep=0pt,minimum size=15pt]
			(bd) at (2,6) {$\vphantom{hy}q_2$};
			\node[circle,draw=black,fill=black!20,inner sep=0pt,minimum size=15pt]
			(dd) at (6,6) {$\vphantom{hy}q_3$};
			\path (ba) edge node[above] {$\delta^0_a$} (aa);
			\path (ba) edge node[above] {$\delta^1_a$} (ca);
			\path (bb) edge node[above] {$\delta^0_a$} (ab);
			\path (bb) edge node[above] {$\delta^1_a$} (cb);
			\path (bc) edge node[above] {$\delta^0_a$} (ac);
			\path (bc) edge node[above] {$\delta^1_a$} (cc);
			\path (ab) edge node[left] {$\delta^0_c$} (aa);
			\path (ab) edge node[left] {$\delta^1_c$} (ac);
			\path (bb) edge node[left] {$\delta^0_c$} (ba);
			\path (bb) edge node[left] {$\delta^1_c$} (bc);
			\path (cb) edge node[left] {$\delta^0_c$} (ca);
			\path (cb) edge node[left] {$\delta^1_c$} (cc);
			\path (bb) edge node[above left] {$\delta^1_{ac}\!\!$} (cc);
			\path (bb) edge node[above left] {$\delta^0_{ac}\!\!$} (aa);
			\path (ad) edge node[left] {$\delta^0_d$} (ac);
			\path (ad) edge node[left] {$\delta^1_d$} (ae);
			\path (bd) edge node[left] {$\delta^0_d$} (bc);
			\path (bd) edge node[left] {$\delta^1_d$} (be);
			\path (dc) edge node[above] {$\delta^0_a$} (cc);
			\path (dc) edge node[above] {$\delta^1_a$} (ec);
			\path (bd) edge node[above] {$\delta^0_a$} (ad);
			\path (bd) edge node[above] {$\delta^1_a$} (cd);
			\path (be) edge node[above] {$\delta^0_a$} (ae);
			\path (be) edge node[above] {$\delta^1_a$} (ce);
			\path (bd) edge node[above left] {$\delta^1_{ad}\!\!$} (ce);
			\path (bd) edge node[above left] {$\delta^0_{ad}\!\!$} (ac);
			\path (dd) edge node[above] {$\delta^0_a$} (cd);
			\path (dd) edge node[above] {$\delta^1_a$} (ed);
			\path (de) edge node[above] {$\delta^0_a$} (ce);
			\path (de) edge node[above] {$\delta^1_a$} (ee);
			\path (cd) edge node[left] {$\delta^0_d$} (cc);
			\path (cd) edge node[left] {$\delta^1_d$} (ce);
			\path (dd) edge node[left] {$\delta^0_d$} (dc);
			\path (dd) edge node[left] {$\delta^1_d$} (de);
			\path (ed) edge node[left] {$\delta^0_d$} (ec);
			\path (ed) edge node[left] {$\delta^1_d$} (ee);
			\path (dd) edge node[above left] {$\delta^1_{ad}\!\!$} (ee);
			\path (dd) edge node[above left] {$\delta^0_{ad}\!\!$} (cc);
			\node[below left] at (ab) {$\bot\;$};
			\node[above right] at (ee) {$\;\top$};
		\end{scope}
		\begin{scope}[shift={(-4,6.25)}]
			\node[right] at (9,-7.5) {$\cell{H}[\emptyset]=\{v_1,\dots, v_8\}$, $\cell{H}[a]=\{t_1,t_2,t_6,t_7,t_9\}$};
			\node[right] at (9,-6.7) {$\cell{H}[c]=\{t_3,t_4\}$, $\cell{H}[d]=\{t_5,t_8,t_{10}\}$};
			\node[right] at (9,-5.9) {$\cell{H}[\loset{a\\c}]=\{q_1\}$};
			\node[right] at (9,-5.1) {$\cell{H}[\loset{a\\d}]=\{q_2,q_3\}$};
			\node[right] at (9,-4.3) {$\bot_{\cell{H}}=\{t_3\}$, $\top_{\cell{H}}=\{v_8\}$};
		\end{scope}
		\begin{scope}[shift={(12,1.5)}, x=1.3cm, y=1.15cm, scale=.9]
			\filldraw[color=black!15] (0,0)--(2,0)--(2,2)--(0,2)--(0,0);
			\filldraw[color=black!15] (0,2)--(0,4)--(4,4)--(4,2)--(0,2);
			\filldraw (0,0) circle (0.05);
			\filldraw (2,0) circle (0.05);
			\filldraw (0,2) circle (0.05);
			\filldraw (2,2) circle (0.05);
			\filldraw (0,4) circle (0.05);
			\filldraw (4,2) circle (0.05);
			\filldraw (4,4) circle (0.05);
			\filldraw (2,4) circle (0.05);
			\path[line width=.5] (0,0) edge node[below, black] {$\vphantom{b}a$} (1.95,0);
			\path[line width=.5] (0,2) edge node[left, black] {$\vphantom{b}d$} (0,3.95);
			\path[line width=.5] (0,2) edge (1.95,2);
			\path[line width=.5] (2,2) edge node[pos=.6, below, black] {$\vphantom{bg}a$} (3.95,2);
			\path[line width=.5] (2,4) edge (3.95,4);
			\path[line width=.5] (0,0) edge node[pos=.6, left, black] {$\vphantom{bg}c$} (0,1.95);
			\path[line width=.5] (2,0) edge (2,1.95);
			\path[line width=.5] (0,4) edge (1.95,4);
			\path[line width=.5] (2,2) edge (2,3.95);
			\path[line width=.5] (4,2) edge (4,3.95);
			\node[left] at (0,0.9) {$\bot$};
			\node[above] at (4,4) {$\top$};
			
			\node[blue,centered] at (0,-0.2) {$v_1$};
			\node[centered, green!50!black] at (1,0.15) {$t_1$};
			\node[blue,centered] at (2,-0.2) {$v_3$};
			\node[centered,blue] at (1.8,2.2) {$v_4$};
			\node[centered,blue] at (-0.2,2) {$v_2$};
			
			\node[centered,blue] at (-0.2,4) {$v_5$};
			\node[centered,blue] at (4.2,2) {$v_6$};
			\node[centered,blue] at (2,4.2) {$v_7$};
			\node[centered,blue] at (4.2,4) {$v_8$};
			
			\node[centered, green!50!black] at (0.2,1.1) {$\vphantom{bg}t_3$};
			\node[centered, green!50!black] at (1.8,1.1) {$\vphantom{bg}t_4$};
			\node[centered, green!50!black] at (1,1.75) {$t_2$};
			\node[centered, green!50!black] at (0.15,3) {$t_5$};
			\node[centered, green!50!black] at (3,2.15) {$t_7$};
			\node[centered, green!50!black] at (3,3.85) {$t_9$};
			\node[centered, green!50!black] at (1.1,3.8) {$\vphantom{bg}t_6$};
			\node[centered, green!50!black] at (2.2,3.1) {$\vphantom{bg}t_8$};
			\node[centered, green!50!black] at (3.8,3.1) {$\vphantom{bg}t_{10}$};
			\node[centered] at (1,1) {$q_1$};
			\node[centered] at (1,3) {$q_2$};
			\node[centered] at (3,3) {$q_3$};
		\end{scope}
	\end{tikzpicture}
	\caption{A two-dimensional HDA $\hda{H}$ on $\Sigma=\{a, c, d\}$, see Example~\ref{ex:hda}.}
	\label{fi:abcube}
\end{figure}
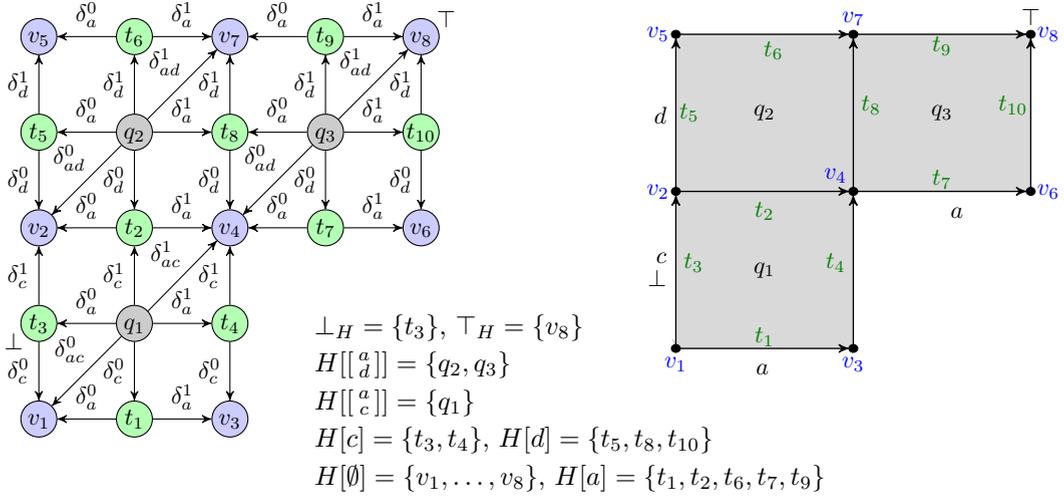

A \emph{path} in an HDA $\hda{H}$ is a sequence
$\alpha=(q_0, \phi_1, q_1,\dots, \phi_n, q_n)$
consisting of cells $q_i\in \cell{H}$ and symbols $\phi_i$ which indicate face map types:
for every $i=1,\dots, n$, $(q_{i-1}, \phi_i, q_i)$ is either
\begin{itemize}
\item $(\delta^0_A(q_i), \arrO{A}, q_i)$ for $A\subseteq \ev(q_i)$ (an \emph{upstep}) or
\item $(q_{i-1}, \arrI{A}, \delta^1_A(q_{i-1}))$ for $A\subseteq \ev(q_{i-1})$ (a \emph{downstep}).
\end{itemize}
Downsteps terminate events, following upper face maps,
whereas upsteps start events by following inverses of lower face maps.

The \emph{source} and \emph{target} of $\alpha$ as above are $\src(\alpha)=q_0$ and $\tgt(\alpha)=q_n$,
and $\alpha$ is \emph{accepting} if $\src(\alpha)\in \bot$ and $\tgt(\alpha)\in \top$.
Paths $\alpha$ and $\beta$ may be \emph{concatenated} to $\alpha*\beta$ if $\tgt(\alpha)=\src(\beta)$.

The \emph{event ipomset} $\ev(\alpha)$ of a path $\alpha$ is defined recursively as follows:
\begin{itemize}
\item if $\alpha=(q)$, then
  $\ev(\alpha)=\id_{\ev(q)}$;
\item if $\alpha=(q\arrO{A} p)$, then
  $\ev(\alpha)=\starter{\ev(p)}{A}$;
\item if $\alpha=(p\arrI{B} q)$, then
  $\ev(\alpha)=\terminator{\ev(p)}{B}$;
\item if $\alpha=\alpha_1*\dotsm*\alpha_n$ is a concatenation, then $\ev(\alpha)=\ev(\alpha_1)*\dotsm*\ev(\alpha_n)$.
\end{itemize}

\begin{example}
	\label{ex:paths}
	The HDA $\hda{H}$ of Example~\ref{ex:hda} (Figure~\ref{fi:abcube}) admits several accepting paths,
	for example
	$t_3\arrO{a} q_1\arrI{c} t_2 \arrO{d} q_2 \arrI{a} t_8 \arrO{a} q_3 \arrI{ad} v_8$.
	Its event ipomset is 
	\begin{equation*}
		\starter{\!\loset{a\\c}}{a}\,* \terminator{\loset{a\\c}}{c}\,*
		\starter{\!\loset{a\\d}}{d}\,* \terminator{\loset{a\\d}}{a}\,*
		\starter{\!\loset{a\\d}}{a}\,* \terminator{\loset{a\\d}}{ad} =
		\left[\vcenter{\hbox{%
				\begin{tikzpicture}[y=1.2cm]
                                  \node (a) at (0.4,1.4) {$a$};
                                  \node at (.25,.7) {$\ibullet\vphantom{d}$};
                                  \node (c) at (0.4,0.7) {$c\vphantom{d}$};
                                  \node (b) at (1.8,1.4) {$a$};
                                  \node (d) at (1.8,0.7) {$d$};
                                  \path (a) edge (b);
                                  \path (c) edge (d);
                                  \path (c) edge (b);
                                  \path[densely dashed, gray] (b) edge (d) (a) edge (d) (a) edge (c);
				\end{tikzpicture}
		}}\right]
	\end{equation*}
	which induces a sparse step decomposition.
\end{example}


The \emph{language} of an HDA $\hda{H}$ is
\begin{equation*}
  \Lang(\hda{H}) = \{\ev(\alpha)\mid \alpha \text{ accepting path in } \hda{H}\} \subseteq \iPoms.
\end{equation*}
Languages of HDAs are closed under subsumption \cite{DBLP:journals/lmcs/FahrenbergJSZ24}:
whenever $P\subseq Q\in L(\hda{H})$, then also $P\in L(\hda{H})$.
This motivates the following definition of language.

For a subset $\mcal S\subseteq \iiPoms$ we let
\begin{equation*}
  \mcal S\down=\{P\in \iiPoms\mid \exists Q\in \mcal S: P\subseq Q\}
\end{equation*}
denote its (downward) closure under subsumptions.
%
%
A \emph{language} is a subset $L\subseteq \iiPoms$ for which $L\down=L$.

A language is \emph{regular} if it is the language of a finite HDA.
A language is \emph{rational} if it is constructed from $\emptyset$, $\{\id_\emptyset\}$ and discrete ipomsets
using $\cup$, $*$ and $^+$ (Kleene plus)~\cite{DBLP:journals/lmcs/FahrenbergJSZ24}.
These operations have to take subsumption closure into account; in particular,
\begin{equation*}
	L_1 * L_2 = \{P*Q\mid P\in L_1, Q\in L_2\}\down.
\end{equation*}

\begin{theorem}[\cite{DBLP:journals/lmcs/FahrenbergJSZ24}]
  \label{th:kleene}
  A language is regular if and only if it is rational.
\end{theorem}

The \emph{width} of a language $L$ is $\wid(L)=\sup\{\wid(P)\mid P\in L\}$.
\begin{lemma}[\cite{DBLP:journals/lmcs/FahrenbergJSZ24}]
	\label{lem:finitewidth}
	Any regular language has finite width.
\end{lemma}

\subsection{ST-automata} 
\label{sse:st-a}

We now introduce ST-automata and the translation from HDAs to these structures.
Recall that $\Omega=\St\cup \Te$ denotes the set of starters and terminators over $\Sigma$.

\begin{definition}
  \label{de:staut}
  An \emph{ST-automaton} is a structure $A=(Q, E, I, F, \lambda)$
  consisting of sets $Q$, $E\subseteq Q\times \Omega\times Q$, $I, F\subseteq Q$,
  and a function $\lambda: Q\to \square$ such that
  for all $(q, \ilo{S}{U}{T}, r)\in E$, $\lambda(q)=S$ and $\lambda(r)=T$.
\end{definition}

This is thus a plain automaton (finite or infinite) over $\Omega$,
with an additional labeling of states with conclists that is consistent with the labeling of edges.
(But note that the alphabet $\Omega$ is infinite.)

\begin{remark}
  \label{re:staomega}
  Equivalently, an ST-automaton may be defined as a directed multigraph $G$
  together with a graph morphism $\ev: G\to \bOmega$
  and initial and final states $I$ and $F$.
  This definition would be slightly more general than the one above,
  given that it allows for multiple edges with the same label between the same pair of states.
\end{remark}

A \emph{path} in an ST-automaton $A$ is defined as usual, as
an alternating sequence $\pi=(q_0, e_1, q_1,\dots, e_n, q_n)$
of states $q_i$ and transitions $e_i$ such that $e_i=(q_{i-1}, P_i, q_i)$ for  every $i=1,\dots, n$
and some sequence $P_1,\dots, P_n\in \Omega$.
The path is \emph{accepting} if $q_0\in I$ and $q_n\in F$.
The \emph{label} of $\pi$ as above is
$\ell(\pi)=[\id_{\lambda(q_0)} P_1 \id_{\lambda(q_1)}\dots P_n \id_{\lambda(q_n)}]_\sim$.
That is, to compute $\ell(\pi)$ we collect labels of states and transitions,
but then we map the so-constructed coherent word to its step sequence.

The \emph{language} of an ST-automaton $A$ is
\begin{equation*}
  \Lang(A) = \{\ell(\pi)\mid \pi \text{ accepting path in } A\} \subseteq \Cohsim.
\end{equation*}
Contrary to languages of HDAs, languages of ST-automata may not be closed under subsumption, see below.

\subsection{From HDAs to ST-automata}
\label{sse:hda-sta-translations}

We now define the translation from HDAs to ST-automata.
In order to relate it to their languages,
we extend the pair of functors $\Phi: \iPoms\leftrightarrows \Cohsim: \Psi$
to the power sets the usual way:
\begin{equation*}
  \Phi(\mcal{S}) = \{\Phi(P)\mid P\in \mcal{S}\}, \qquad
  \Psi(\mcal{S}) = \{\Psi(w)\mid w\in \mcal{S}\}.
\end{equation*}

To a given HDA $\hda{H}=(\cell{H}, \bot, \top)$ we associate an ST-automaton $\ST(\hda{H})=(Q, E, I, F, \lambda)$ as follows:
\begin{itemize}
\item $Q=\cell{H}$, $I=\bot$, $F=\top$, $\lambda=\ev$, and
\item $E=\{(\delta_A^0(q), \starter{\ev(q)}{A}, q)\mid A\subseteq \ev(q)\}
  \cup \{(q, \terminator{\ev(q)}{A}, \delta_A^1(q))\mid A\subseteq \ev(q)\}$.
\end{itemize}
That is, the transitions of $\ST(\hda{H})$ precisely mimic the starting and terminating of events in $\hda{H}$.
(Note that lower faces in $\hda{H}$ are inverted to get the starting transitions.)

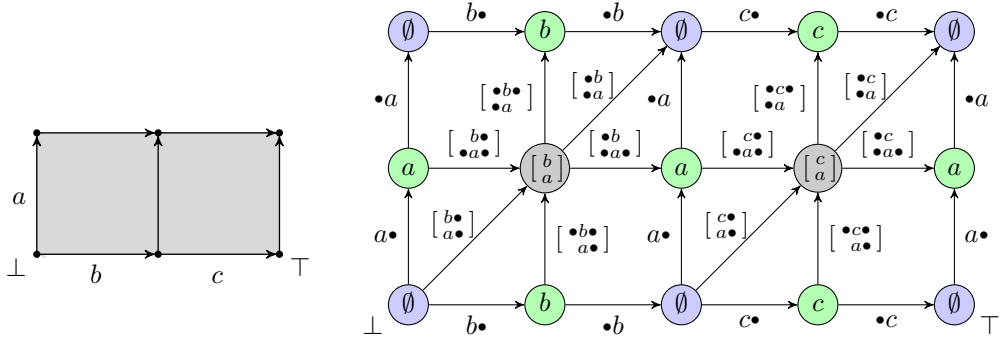
\begin{figure}[tbp]
	\centering
	\begin{tikzpicture}[x=.7cm, y=.62cm]
		\begin{scope}[shift={(0,0)}, x=.8cm, y=.8cm]
			\filldraw[color=black!15] (0,2)--(0,4)--(4,4)--(4,2)--(0,2);
			\filldraw (0,2) circle (0.05);
			\filldraw (2,2) circle (0.05);
			\filldraw (0,4) circle (0.05);
			\filldraw (4,2) circle (0.05);
			\filldraw (4,4) circle (0.05);
			\filldraw (2,4) circle (0.05);
			\path[line width=.5] (0,2) edge node[left, black] {$\vphantom{b}a$} (0,3.95);
			\path[line width=.5] (0,2) edge node[pos=.5, below, black] {$b$} (1.95,2);
			\path[line width=.5] (2,2) edge node[pos=.5, below, black] {$\vphantom{bg}c$} (3.95,2);
			\path[line width=.5] (2,4) edge (3.95,4);
			\path[line width=.5] (0,4) edge   (1.95,4);
			\path[line width=.5] (2,2) edge (2,3.95);
			\path[line width=.5] (4,2) edge (4,3.95);
			\node[left] at (0,1.8) {$\bot$};
			\node[right] at (4,1.8) {$\top$};
			
		\end{scope}
		\begin{scope}[shift={(7,-4.3)}, x=0.9cm, y=0.9cm]
			\node[circle,draw=black,fill=blue!20,inner sep=0pt,minimum size=15pt]
			(ac) at (0,4) {$\vphantom{hy}\emptyset$};
			\node[circle,draw=black,fill=blue!20,inner sep=0pt,minimum size=15pt]
			(cc) at (4,4) {$\vphantom{hy}\emptyset$};
			\node[circle,draw=black,fill=blue!20,inner sep=0pt,minimum size=15pt]
			(ae) at (0,8) {$\vphantom{hy}\emptyset$};
			\node[circle,draw=black,fill=blue!20,inner sep=0pt,minimum size=15pt]
			(ec) at (8,4) {$\vphantom{hy}\emptyset$};
			\node[circle,draw=black,fill=blue!20,inner sep=0pt,minimum size=15pt]
			(ce) at (4,8) {$\vphantom{hy}\emptyset$};
			\node[circle,draw=black,fill=blue!20,inner sep=0pt,minimum size=15pt]
			(ee) at (8,8) {$\vphantom{hy}\emptyset$};
			\node[circle,draw=black,fill=green!30,inner sep=0pt,minimum size=15pt]
			(bc) at (2,4) {$\vphantom{hy}b$};
			\node[circle,draw=black,fill=green!30,inner sep=0pt,minimum size=15pt]
			(ad) at (0,6) {$\vphantom{hy}a$};
			\node[circle,draw=black,fill=green!30,inner sep=0pt,minimum size=15pt]
			(be) at (2,8) {$\vphantom{hy}b$};
			\node[circle,draw=black,fill=green!30,inner sep=0pt,minimum size=15pt]
			(cd) at (4,6) {$\vphantom{hy}a$};
			\node[circle,draw=black,fill=green!30,inner sep=0pt,minimum size=15pt]
			(de) at (6,8) {$\vphantom{hy}c$};
			\node[circle,draw=black,fill=green!30,inner sep=0pt,minimum size=15pt]
			(dc) at (6,4) {$\vphantom{hy}c$};
			\node[circle,draw=black,fill=green!30,inner sep=0pt,minimum size=15pt]
			(ed) at (8,6) {$\vphantom{hy}a$};
			\node[circle,draw=black,fill=black!20,inner sep=0pt,minimum size=15pt]
			(bd) at (2,6) {$\vphantom{hy}\loset{b \\ a}$};
			\node[circle,draw=black,fill=black!20,inner sep=0pt,minimum size=15pt]
			(dd) at (6,6) {$\vphantom{hy}\loset{c \\ a}$};
			\path (ac) edge node[below] {$b\ibullet$} (bc);
			\path (bc) edge node[below] {$\ibullet b$} (cc);
			\path (ac) edge node[left] {$a \ibullet$} (ad);
			\path (ad) edge node[left] {$\ibullet a$} (ae);
			\path (bc) edge node[right] {$\loset{\ibullet b \ibullet \\ \hphantom{\ibullet} a \ibullet }$} (bd);
			\path (bd) edge node[left] {$\loset{\ibullet b \ibullet \\ \ibullet a \hphantom{\ibullet}}$} (be);
			\path (cc) edge node[below] {$c\ibullet$} (dc);
			\path (dc) edge node[below] {$\ibullet c$} (ec);
			\path (ad) edge node[above] {$\loset{\hphantom{\ibullet} b \ibullet \\ \ibullet a \ibullet}$} (bd);
			\path (bd) edge node[above] {$\loset{\ibullet b \hphantom{\ibullet}\\ \ibullet a \ibullet}$} (cd);
			\path (ae) edge node[above] {$b \ibullet$} (be);
			\path (be) edge node[above] {$\ibullet b$} (ce);
			\path (bd) edge node[above left=-0.15cm] {$\loset{\ibullet b \\ \ibullet a}$} (ce);
			\path (ac) edge node[above left=-0.15cm] {$\loset{b \ibullet \\ a \ibullet}$} (bd);
			\path (cd) edge node[above] {$\loset{\hphantom{\ibullet}c \ibullet \\ \ibullet a \ibullet}$} (dd);
			\path (dd) edge node[above] {$\loset{\ibullet c \hphantom{\ibullet}\\ \ibullet a \ibullet}$} (ed);
			\path (ce) edge node[above] {$c \ibullet $} (de);
			\path (de) edge node[above] {$\ibullet c$} (ee);
			\path (cc) edge node[left] {$a \ibullet$} (cd);
			\path (cd) edge node[left] {$\ibullet a$} (ce);
			\path (dc) edge node[right] {$\loset{\ibullet c \ibullet \\ \hphantom{\ibullet} a \ibullet}$} (dd);
			\path (dd) edge node[left] {$\loset{\ibullet c \ibullet \\   \ibullet a \hphantom{\ibullet}}$} (de);
			\path (ec) edge node[right]{$a \ibullet $} (ed);
			\path (ed) edge node[right] {$\ibullet a$} (ee);
			\path (dd) edge node[above left=-0.15cm] {$\loset{ \ibullet c \\ \ibullet a }$} (ee);
			\path (cc) edge node[above left=-0.15cm] {$\loset{c \ibullet \\ a \ibullet}$} (dd);
			\node[below left] at (ac) {$\bot\;\;$};
			\node[below right] at (ec) {$\;\;\top$};
		\end{scope}
	\end{tikzpicture}
	\caption{Two-dimensional HDA $\hda{H}$ (left) and corresponding ST-automaton $\ST(\hda{H})$ (right).}
	\label{fi:hdast}
\end{figure}

\begin{example}
  Figure \ref{fi:hdast} shows an HDA $\hda{H}$ with $\Lang(\hda{H})=\{bc\}$
  together with its translation to an ST-automaton $\ST(\hda{H})$.
\end{example}

\begin{theorem}
  \label{th:HDA-STA}
  For any HDA $\hda{H}$, $\Lang(\ST(\hda{H}))=\Phi(\Lang(\hda{H}))$.
\end{theorem}

\begin{proof}
  For identities note that a path with a single cell $q$ is accepting in 
  $\hda{H}$ if and only if it is accepting in $\ST(\hda{H})$,
  and $\Phi(\id_{\ev(q)}) = [\id_{\lambda(q)}]_\sim$.
  Now let $w = P_1\dots P_m \in \Lang(\ST(\hda{H}))$ be a non-identity.
  By definition, there exists $\pi=(q_0, e_1, q_1,\dots, e_n, q_n)$
  where $e_i = (q_{i-1}, P'_i, q_i)$, $P'_i \in \Omega$
  such that $\id_{\lambda(q_{0})}P'_1 \id_{\lambda(q_1)} \dots P'_n \id_{\lambda(q_n)} \sim w$.
  This means that $P'_1* \dots *P'_n$ is a decomposition of some $P \in \Lang(\hda{H})$,
  hence $w \in \Phi(\Lang(\hda{H}))$.
		
  For the converse, let $w = P_1\dots P_m \in \Phi(\Lang(\hda{H}))$.
  Let $P'_1*\dots *P'_n$ be the sparse step decomposition of $P = P_1 * \dots * P_m$.
  We have $P'_1\dots P'_n \sim w$.
  In addition, there exists an accepting path $\alpha = \beta_1 * \dots * \beta_n$  in $\hda{H}$ such that $\ev(\beta_i) = P'_i$.
  By construction there exists an accepting path
  $\pi = (\src(\beta_1),e_1,\tgt(\beta_1),\dots,e_n,\tgt(\beta_n))$ in $\ST(\hda{H})$ where $e_i = (\src(\beta_i),P'_i,\tgt(\beta_i))$.
  We have $\ell(\pi) \sim w$. 
\end{proof}

\section{Monadic Second-Order Logic for HDAs}
\label{sec:MSO}

We now consider the isomorphism between ipomsets and step sequences from a logical perspective.
In this section, we define monadic second-order logic (MSO) over ipomsets and words over starters and terminators, and show that an MSO formula  over ipomsets can be turned into an MSO formula over $\Omega^*$, and vice versa, with equivalent languages up to~$\Psip$.
By combining this equivalence with various tools -- such as Kleene theorems, ST-automata, and others -- we obtain a Büchi-Elgot-\!Trakhtenbrot theorem relating the expressive power of MSO over ipomsets to that of HDAs.

More precisely, since $\Omega = \St \cup \Te$ -- the set of starters and terminators -- is infinite and in order to restrict to a finite alphabet, we actually work with MSO over words formed of \emph{width-bounded} starters and terminators, for some width bound fixed by the context.
This is not an issue, thanks to Lemma~\ref{lem:finitewidth}.
For $k\in \Nat$, we denote by $\iiPoms_{\le k}=\{P\in \iiPoms\mid \wid(P)\le k\}$ and, for $L\subseteq \iiPoms$, $L_{\le k}= L \cap \iiPoms_{\le k}$.
In particular, $\St_{\le k}= \St\cap \iiPoms_{\le k}$ and $\Te_{\le k}= \Te\cap \iiPoms_{\le k}$
denote the finite sets of starters and terminators of width at most $k$;
further,
$\Coh_{\le k} = \Coh \cap \Omega^*_{\le k}$.

\subsection{MSO}

Monadic second-order logic is an extension of first-order logic  allowing to quantify existentially  and universally 
over elements as well as subsets of the domain of the structure.
It uses second-order variables $X, Y, \dots$ interpreted as subsets of the domain in addition to the first-order variables $x, y, \dots$ interpreted as elements of the domain of the structure, and a new binary predicate $x \in X$ interpreted commonly.
We refer the reader to \cite{Thomas97a} for more details about MSO.
In this work, we interpret MSO over $\iiPoms$, referred to as $\MSOP$, and over words in $\Omegaswb^*$ for some fixed width bound $k$, referred to as $\MSOW$. 
In this section, we assume $\Sigma$ to be finite.


For $\MSOP$, we consider the signature $\Sigp = \{{<}, {\evord},$ $(a)_{a\in \Sigma},\srci,\tgti\}$ where $<$ and $\evord$ are 
binary relation symbols and the $a$'s, $\srci$ and $\tgti$ are unary predicates over first-order variables. 
We associate to every ipomset $(P, {<}, {\evord}, S, T, \lambda)$  the relational structure $(P ; {<}; {\evord}; (a)_{a \in \Sigma}; \srci; \tgti)$
where $<$ and $\evord$ are interpreted as the orderings $<$ and $\evord$ over $P$, and $a(x)$, $\srci(x)$ and $\tgti(x)$ hold respectively if $\lambda(x) = a$, $x \in S$ and $x \in T$.
The well-formed $\MSOP$ formulas are built using the following grammar: 
\begin{align*}
  \psi ::={} &a(x), a \in \Sigma ~\vert ~ \srci(x) ~\vert~ \tgti(x)  ~\vert~  x < y ~ \vert ~ x \evord y ~\vert~ x \in X ~\vert~ 
  \exists x.\, \psi ~\vert~ 
   \exists X.\, \psi 
   ~\vert~ \psi_1 \vee \psi_2~\vert~ \neg \psi
\end{align*}

In order to shorten formulas  we use several notations and shortcuts such as
$\psi_1 \lor \psi_2$, $\psi_1\implies\psi_2$, $\forall x.\psi$ or $\forall X.\psi$
which are defined as usual, and the equality predicate $x = y$ as
$\lnot (x < y) \land \lnot (y > x) \land \lnot (x \evord y) \land \lnot (y \evord x)$.
We also define the direct successor relation as
\begin{equation*}
  x \psucc y\eqdef x < y \wedge \neg(\exists z. x < z < y).
\end{equation*}

Let $\psi(x_1,\dots,x_n,X_1,\dots,X_m)$  be an $\MSOP$ formula  whose free variables are $x_1,\dots,x_n,$ $X_1,\dots,X_m$ and let $P \in \iiPoms$. 
A pair of functions $\nu = (\nu_1,\nu_2)$, where $\nu_1 \colon \{x_1,\dots,x_n\} \to P$ and $\nu_2 \colon \{X_1,\dots,X_m\} \to 2^P$, is called a \emph{valuation} or an \emph{interpretation}.
We write $P \models_\nu \psi$ or, by a slight abuse of notation, $P\models \psi\big(\nu(x_1),\dots,\nu(x_n),\nu(X_1),\dots,\nu(X_m)\big)$,  if $\psi$ holds when $x_i$ and $X_j$ are interpreted as $\nu(x_i)$ and $\nu(X_j)$.

We say that a relation $R \subseteq 
P^n \times (2^P)^m$ is \emph{$\MSOP$-definable} if there exists an $\MSOP$ formula	$\psi(x_1 ,\dots , x_n , X_1 , \dots, X_m )$ which is satisfied  if and only if the interpretation of the free variables $(x_1 ,\dots , x_n , X_1 , \dots, X_m)$ is a tuple of $R$.
A \emph{sentence} is a formula without free variables.
In this case no valuation is needed.
Given an $\MSOP$ sentence $\psi$,	we define $L(\psi) = \{ P \in \iiPoms \mid P \models \psi\}$ and $L(\psi)_{\le k} = L(\psi) \cap \iiPoms_{\le k}$.
A set
$L\subseteq \iiPoms$ is $\MSOP$-definable if and only if there exists an
$\MSOP$ sentence $\psi$ over $\mathcal{S}$ such that $L = L(\psi)$. 

\begin{example}\label{ex:notwb}
	Let $\phi = \exists x\, \exists y.\, a(x) \wedge b(y) \wedge \lnot (x < y) \wedge \lnot (y < x)$.
	That is, there are at least two concurrent events, one labelled $a$ and the other $b$.
	$L(\phi)$
	is not width-bounded, as $\phi$ is satisfied, among others, by any conclist which contains at least one $a$ and one $b$.
	Nor is it closed under subsumption, given that $\loset{a\\b}\models \phi$ but $a b, b a\not\models \phi$.
	Note, however, that the width-bounded closure $L(\phi)_{\le k}\down$ \emph{is} a regular language for any $k$.
\end{example}

For $\MSOW$, we consider the signature $\Sigw = \{{<}, (\discrete)_{\discrete \in \Omegaswb}\}$ where $<$ is a binary relation symbol and $D$ (for discrete ipomset) is an unary predicate over first-order variables.
Note that, for a fixed $k$, $\Omegaswb$ is finite.
A word  $w$ of $\Omegaswb^*$ can be seen as $(W, {<}, \lambda\colon W \to \Omegaswb)$:
a totally ordered finite set $W$ labelled by $\Omegaswb$.
Its relational structure is $(W; {<}; (D)_{\discrete \in \Omegaswb})$.
First-order variables range over $W$ and second-order variables over $2^W$, $<$ is interpreted as the ordering $<$ over $W$ and $D(x)$ holds if $\lambda(x) = D$.
The well-formed $\MSOW$ formulas for a fixed width bound $k$ (understood from context) are built using the following grammar: 
\begin{align*}
  \psi ::={} &D(x), D \in \Omegaswb ~\vert ~  x < y  ~\vert~ x \in X  ~\vert~ 
  \exists x.\, \psi 
  ~\vert~ \exists X.\, \psi 
  ~\vert~ \psi_1 \vee \psi_2~\vert~ \neg \psi
\end{align*}

Interpretation, definability and satisfaction in $\MSOW$ are defined analogously to $\MSOP$.
Recall that we interpret $\MSOW$ over $\Omegaswb^*$, that is words over starters and terminators of width bounded by $k$.
Thus, given an $\MSOW$ sentence $\psi$, $L(\psi) = \{ w \in \Omegaswb^* \mid w \models \psi\}$.

\begin{example}
Given a word $w = P_1 \ldots P_n \in \Omegaswb^\ast$, the following formula 
is satisfied precisely by $w$:
\begin{equation*}
	\exists y_1,\dots,y_n. \bigwedge_{1 \le i \le n} P_i(y_i) \land 
	\bigwedge_{1 \le i < n} y_i \psucc y_{i+1} 
	\land \forall y. \bigvee_{1 \le i \le n} y = y_i\
\end{equation*}
\end{example}

Note that $\MSOW$ sentences may be satisfied by non-coherent words.
We say that an $\MSOW$ sentence $\phi$ is \emph{$\sim$-invariant}
if for all words $w,w' \in \Coh$ such that $w \sim w'$, we have
$w \models \phi$ if and only if $w' \models \phi$.
Then, we define $\widetildeL(\phi) = \{[s]_\sim \in \Cohsim \mid s \models \phi\}$.

In this section, we show that there are effective translations between $\MSOP$
and $\MSOW$.
Recall that the mapping $\bar\Phi \colon \iiPoms \to \Coh$ maps an ipomset $P$ to its sparse step decomposition $\Phip(P)$
and that $\SCoh = \Phip(\iiPoms)$ denotes the set of all \emph{sparse} coherent words over $\Omega$.
The next two subsections are devoted to the proof of the following theorem:
\begin{theorem}
	\label{th:MSOeq}
	\begin{enumerate}
		\item For every sentence $\phi \in \MSOP$ and every $k$, there exists
		a sentence $\widehat \phi \in \MSOW$ over $\Omegaswb^*$ such that
		$L(\widehat \phi) = \Psip^{-1}(L(\phi)_{\le k})$.
		\item For every $k$ and every  sentence $\phi \in \MSOW$ over $\Omegaswb^*$,
		there exists a sentence $\overline \phi \in \MSOP$
		such that $L(\overline \phi) = \Psip(L(\phi) \cap \SCoh)$.
	\end{enumerate}
\end{theorem}

In addition, the constructions are effective. 
By passing through ST-automata and classical automata, we obtain the following corollary.

\begin{corollary}
  \label{cor:MSOdef}
  Let $L \subseteq \iiPoms_{\le k}$.
  The following are equivalent:
  \begin{enumerate}
  \item $L$ is $\MSOP$-definable;
  \item $\Phip(L)$ is $\MSOW$-definable;
  \item $\Phi(L)$ is $\MSOW$-definable.
  \end{enumerate}
\end{corollary}

\begin{remark}
    The reader familiar with MSO transductions (see e.g.\ \cite{DBLP:books/daglib/0030804}) may notice that the next two subsections essentially show that $\Psip$ and  $\Phip$ can be defined through MSO transductions (see in particular Lemmas~\ref{lem:baserel}, \ref{lem:compare} and \ref{lem:StTe(x)Def}).
\end{remark}

\subsection{From $\iiPoms_{\le k}$ to $\Coh_{\le k}$}

Here, we prove the first item of Theorem~\ref{th:MSOeq}.
For now, we restrict ourselves to words without any occurrence of the empty 
ipomset $\id_{\emptyset}$.
Our goal is to prove the following:

\begin{lemma}
	\label{lem:msodec}
	For every $\phi \in \MSOP$ and every $k$, there exists 
	$\widehat \phi \in \MSOW$ such that for all
	$w \in (\Omegaswb \setminus \{\id_\emptyset\})^+ $, we have
	$w \models \widehat \phi$ if and only if $w \in \Coh$ and 
	$\Psip(w) \models \phi$.
\end{lemma}

Prior to proving the lemma, we introduce some notation.
We want  a word $P_1\dots P_n$ of $(\Omegaswb \setminus \{\id_\emptyset\})^+$ to satisfy $\widehat \phi$ if and only if the gluing composition $P = P_1 \comp \cdots \comp P_n$ is a model for $\phi$.
Thus $\widehat \phi$ must accept only coherent words. This is $\MSOW$-definable by:
\[
\texttt{Coh}_k  \eqdef \exists z\, \forall x\, \forall y.\, x \psucc y \implies \!\!\bigvee_{D_1 D_2 \in \Coh \cap \smash{\Omegaswb^2}}\!\! D_1(x) \land D_2(y).
\]
That is, the word is non-empty ($\exists z$) and discrete ipomsets of $\Omegaswb$ at consecutive positions $x$ and $y$ may be glued.

Hence,  $\widehat{\phi}$ will be the conjunction of $\texttt{Coh}_k$ and an $\MSOW$ formula $\phi'$ which we will build by induction on $\phi$.
To construct $\phi'$, the intuition is that, given an ipomset $P$ and a coherent word $w = P_1\dots P_n \in (\Omegaswb \setminus \{\id_\emptyset\})^+$ such that $\Psip(w) =P$, an event $e$ of $P$  may appear in several consecutive  $<$-positions $1 \leq \ell_1,\dots,\ell_m \leq n$ in $w$ and, in each $<$-position $\ell_j$, it occurs once in some $\evord$-position $i_j$.
The goal in $\phi'$ is to select some convenient interval of pairs $(\ell_j,i_j)$ when $\phi$ selects $e$.

More formally, let $w = P_1\dots P_n \in \Coh_{\le k}$ and $P= P_1 * \dots * P_n$.
Let $E = \{1,\ldots,n\} \times \{1,\ldots,k\}$.
Our construction is built on a partial function $\evt\colon E \to P$ defined as follows:
if $P_\ell$ consists of events $e_1 \evord \cdots \evord e_r$, then for every $i \le r$, $\evt(\ell,i) = e_i$.
Our first step towards proving Lemma~\ref{lem:msodec} is to show that all
atomic predicates $x = y$, $x < y$, $a(x)$, etc.\ of $\MSOP$
can be translated into formulas in $\MSOW$:

\begin{lemma}\label{lem:baserel}
	For every $k \in \Nat$ and $1 \le i,j \le k$, one can define
	$\MSOW$ formulas
	$\phidom(x,i)$,
	$(x,i) \newSim (y,j)$,
	$(x,i) < (y,j)$,
	$(x,i) \evord (y,j)$,
	$a(x,i)$,
	$\srci(x,i)$,
	and $\tgti(x,i)$, such that for all $w \in \Coh_{\le k}$ with $\Psip(w) = (P,<_P,\evord_P,S_P,T_P,\lambda_P)$ and for any valuation $\nu$ over $w$:
	\[
	\begin{aligned}
		w & \models_\nu \phidom(x,i)
		&& \quad\text{if and only if}\qquad \evt(\nu(x),i) \text{ is defined} \\
		w & \models_\nu (x,i) \newSim (y,j)
		&& \quad\text{if and only if}\qquad \evt(\nu(x),i) = \evt(\nu(y),j) \\
		w & \models_\nu (x,i) < (y,j)
		&& \quad\text{if and only if}\qquad \evt(\nu(x),i) <_P \evt(\nu(y),j) \\
		w & \models_\nu (x,i) \evord (y,j)
		&& \quad\text{if and only if}\qquad \evt(\nu(x),i) \evord_P \evt(\nu(y),j) \\
		w & \models_\nu a(x,i)
		&& \quad\text{if and only if}\qquad \lambda_P(\evt(\nu(x),i)) = a \\
		w & \models_\nu \srci(x,i)
		&& \quad\text{if and only if}\qquad \evt(\nu(x),i) \in S_P \\
		w & \models_\nu \tgti(x,i)
		&& \quad\text{if and only if}\qquad \evt(\nu(x),i) \in T_P .
	\end{aligned}
	\]
\end{lemma}
\begin{proof}
	The formula $\phidom(x,i)$ simply checks that the discrete ipomset labeling
	$x$ is of size at least $i$:
	\[
		\phidom(x,i) \eqdef \bigvee_{D \in \Omega_{\le k} \setminus \Omega_{\le i-1}} D(x).
	\]
	Let us now define the formula $(x,i) \newSim (y,j)$.
	Let $w = P_1\dots P_n$.
	Notice that an event $e \in P$ may occur in several consecutive
	$P_\ell$'s within $w$.
	So the formula $(x,i) \newSim (y,j)$ is meant to determine
	when $\evt(\ell,i) = \evt(\ell',j)$ for some positions $\ell,\ell'$ within $w$.
	We first consider the case where $\ell' = \ell+1$.
	For all $i',j' \le k$, let $M_{i',j'} = \{\discrete_1\discrete_2 \in \Omega_{\leq k}^2 \mid \evt(1,i') = \evt(2,j') \}$, and
	\[
	\glue i j (x,y) \eqdef x \psucc y \land \!\!\bigvee_{\discrete_1\discrete_2 \in M_{i,j}}\!\! \discrete_1(x) \land \discrete_2(y).
	\]
	Then $w \models_\nu \glue i j (x,y)$ if and only if $\nu(y) = \nu(x)+1$
	and $\evt(\nu(x),i) = \evt(\nu(y),j)$.
	We can then define $\newSim$ as a kind of reflexive transitive closure
	of these relations:
	\begin{multline*}
		(x,i) \newSim (y,j) \eqdef
		\forall X_1, \ldots, X_k.\, 
		\Big( x \in X_i \land \textstyle \bigwedge_{i,j \le k} \forall x,y.\, \\
		x \in X_i \land ((x,i) = (y,j) \lor \glue i j (x,y) \lor \glue j i (y,x)) \implies y \in X_j
		\Big)
		\implies y \in X_j,
	\end{multline*}
	where $(x,i) = (y,j)$ stands for the formula $x = y$ when $i = j$ and false otherwise.
	
	We rely on the formula $(x,i) \newSim (y,j)$ to define $(x,i) < (y,j)$ and $(x,i) \evord (y,j)$.
	
	Notice that given two events $e,e'$ in $P$,
	$e < e'$ iff for all $P_\ell$ and $P_{\ell'}$ in which $e$ and $e'$
	occur, we have $\ell < \ell'$. In other terms, iff for all $(\ell,i')$ and 
	$(\ell',j')$ such that $\evt(\ell,i') = e$ and $\evt(\ell',j') = e'$,
	we have $\ell < \ell'$. This can be expressed as follows:
	\[
	(x,i) < (y,j) \eqdef  \bigwedge_{1 \le i',j' \le k}\! \forall x',y'.
	\big( (x',i') \newSim (x,i) \land (y',j') \newSim (y,j) \big) \implies x' < y'.
	\]
	
	To define $(x,i) \evord (y,j)$, notice that $e \evord e'$ in $P$ when they appear 
	together in one $P_\ell$ in this order, that is, if there exists $\ell,i',j'$
	such that $i'$ is smaller than $j'$ and $\evt(\ell,i') = e$, $\evt(\ell,j') = e'$. This leads us to
	\[
	(x,i) \evord (y,j)
	\eqdef \bigvee_{1 \le i' < j' \le k}\! \exists z\, 
	(z,i') \newSim (x,i) \land (z,j') \newSim (y,j).
	\]
	
	For the unary predicates, we let
	\[
	a(x,i) \eqdef \bigvee_{\discrete \in \Omega_{a,i}} \discrete(x) \, ,
	\qquad
	\srci(x,i) \eqdef \bigvee_{\discrete \in \Omega_{\srci,i}} \discrete(x) \, ,
	\qquad
	\tgti(x,i) \eqdef \bigvee_{\discrete \in \Omega_{\tgti,i}} \discrete(x) \, ,
	\]
	where $\Omega_{a,i}$ (resp.~$\Omega_{\srci,i}$, resp.~$\Omega_{\tgti,i}$)
	is the (finite) set of all $\discrete \in \Omegaswb$ with events 
	$e_1 \evord \cdots \evord e_r$ such that $r \ge i$ and 
	$\lambda_D(e_i) = a$ (resp.~$e_i \in S_\discrete$, resp.~$e_i \in T_\discrete$).
\end{proof}

Let us now see how to use these base formulas in a translation
from $\MSOP$ to $\MSOW$ and prove Lemma~\ref{lem:msodec}.
As mentioned before, $\widehat{\phi}$ will be the conjunction of
$\texttt{Coh}_k$ and an $\MSOW$ formula $\phi'$ which we build
by induction on $\phi$.
Since we proceed by induction, we have to consider formulas $\phi$ that contain
free variables.
We construct $\phi'$ so that its free first-order variables are the same as $\phi$,
and it has second-order variables $X_1,\ldots,X_k$ for every free second-order 
variable $X$ of $\phi$.
In addition, every first-order variable of $\phi'$ is paired to some
$i \in \{1,\ldots,k\}$ by a function $\tau$, given as a parameter in the
translation.
Intuitively, we want to replace $x$ with the pair $(x,\tau(x))$, and $X$
with the union $\bigcup_{1 \le i \le k} \{(x,i) \mid x \in X_i\}$.
The next lemma expresses this more precisely, and Example~\ref{ex:msotohda}
illustrates it.

\begin{lemma}
	For every $\MSOP$ formula $\phi$ and function $\tau$ from
	the free first-order variables of $\phi$ to $\{1,\ldots,k\}$, one can construct
	a formula $\phi'_\tau \in \MSOW$ such that for every $w \in \Coh_{\le k}$
	and $P = \Psip(w)$, 
	\[
	P \models_\nu \phi \quad\text{if and only if}\quad
	w \models_{\nu'} \phi'_\tau
	\]
	for any valuations $\nu$ and $\nu'$ satisfying the following conditions:
	\begin{enumerate}
		\item $\evt(\nu'(x),\tau(x)) = \nu(x)$ and 
		\item $\bigcup_{1 \le i \le k} \{\evt(e,i) \mid e \in \nu'(X_i)\} 
		= \nu(X)$.
	\end{enumerate}
\end{lemma}

\begin{proof}
	We make use of the formulas from Lemma~\ref{lem:baserel}:
	\begin{align*}
		(x = y)'_\tau & \eqdef (x,\tau(x)) \newSim (y,\tau(y)) \\
		(x < y)'_\tau & \eqdef (x,\tau(x)) < (y,\tau(y)) \\
		(x \evord y)'_\tau & \eqdef (x,\tau(x)) \evord (y,\tau(y)) \\
		a(x)'_\tau & \eqdef a(x,\tau(x)) \\
		\srci(x)'_\tau & \eqdef \srci(x,\tau(x)) \\
		\tgti(x)'_\tau & \eqdef \tgti(x,\tau(x)).
	\end{align*}	
	The function~$\tau$ emerges in the case $\phi = \exists x.\, \psi$,
	where we let
	\[
	\phi'_\tau \eqdef \bigvee_{1 \le i \le k} \exists x.\, \phidom(x,i) \land \psi'_{\tau[x \mapsto i]}.
	\]
	For the second-order part, we let
	\begin{align*}
		(\exists X.\, \psi)'_\tau
		& \eqdef \exists X_1,\ldots, X_k.\, \psi'_\tau \land \bigwedge_{1 \le i \le k} \forall x.\, x \in X_i \implies \phidom(x,i) \\
		(x \in X)'_\tau
		& \eqdef \bigvee_{1 \le j \le k} \exists y\, (x,\tau(x)) \newSim (y,j) \land y \in X_j.
	\end{align*}
	Finally, when $\phi$ is $\psi_1 \lor \psi_2$ or $\neg \psi$, then we let $\phi'$ be $\psi'_1 \lor \psi'_2$ or $\neg \psi'$, respectively.
\end{proof}

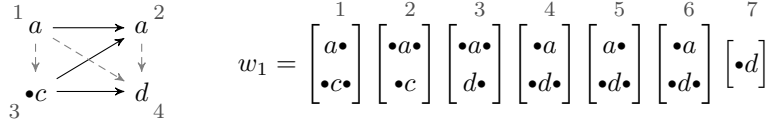
\begin{figure}[tbp]
	\centering
	\begin{tikzpicture}[y=1.2cm]
          \begin{scope}
		\node (a) at (0.4,1.4) {$a$};
		\node (c) at (0.4,0.7) {$\ibullet c\vphantom{d}$};
		\node (b) at (1.8,1.4) {$a$};
		\node (d) at (1.8,0.7) {$d$};
		\path (a) edge (b);
		\path (c) edge (d);
		\path (c) edge (b);
		\path[densely dashed, gray] (b) edge (d) (a) edge (d) (a) edge (c);
		\node[font=\scriptsize, black!70] at (a.north west) {1};
		\node[font=\scriptsize, black!70] at (b.north east) {2};
		\node[font=\scriptsize, black!70] at (c.south west) {3};
		\node[font=\scriptsize, black!70] at (d.south east) {4};
          \end{scope}
		\begin{scope}[shift={(3.5,1.1)},x=.92cm]
			\path[use as bounding box] (-.45,0) -- (7,-.2);
			\foreach \i in {1,2,3,4,5,6} \node[font=\scriptsize, black!70] at (\i, 0.5) {\i};
			\node[font=\scriptsize, black!70] at (6.9, 0.5) {7};
			
			\node at (3.35,-0.1) {
				$w_1 =
				\bigloset{ a \ibullet \\ \ibullet c \ibullet}
				\bigloset{\ibullet a \ibullet \\ \ibullet c}
				\bigloset{\ibullet a \ibullet \\ d \ibullet}
				\bigloset{\ibullet a \\ \ibullet d \ibullet}
				\bigloset{ a \ibullet \\ \ibullet d \ibullet}
				\bigloset{ \ibullet a \\ \ibullet d \ibullet}
				\bigloset{\ibullet d}
				$
			};
		\end{scope}
	\end{tikzpicture}
	\caption{Ipomset and corresponding coherent word.
          On the left, numbers indicate events; on the right, positions.}
	\label{fi:example}
\end{figure}

\begin{example}
	\label{ex:msotohda}
	Figure~\ref{fi:example} displays an ipomset $P$ and the coherent word $w_1 = P_1 \dots P_7$ such that $P_1 * \dots * P_7 = P$.
	Let $e_1,\dots,e_4$ be the events of $P$ labelled respectively by the left $a$, the right $a$, $c$, and $d$ and let $p_1, \dots,p_7$ the positions on $w_1$ from left to right.
	Assume that $P \models_\nu \phi(x,X)$ for some $\MSOP$-formula $\phi$  and the valuation $\nu(x) = e_1$ and $\nu(X) =\{e_2,e_3\}$.
	Then,  $w_1 \models_{\nu'}\phi'_{[x \mapsto 1]}(x,X_1,X_2)$ when, for example, $\nu'(x) = p_2$, $\nu'(X_1) = \{p_6\}$ and $\nu'(X_2) = \{p_3\}$ since this valuation satisfies the invariant property. 
	For $\newSim$ we have 
	$(p_1,1) \newSim \dots \newSim (p_4,1)$, $(p_1,2) \newSim (p_2,2)$, $(p_3,2) \newSim \dots \newSim (p_6,2) \newSim (p_7,1)$ and $(p_5,1) \newSim (p_6,1)$.
	In particular $(p_1,1)\not\newSim (p_5,1)$ since neither $\glue 1 1 (p_4,p_5)$ nor $\glue 2 1 (p_4,p_5)$ hold.
\end{example}

We have now proven Lemma~\ref{lem:msodec}. The first assertion of Theorem~\ref{th:MSOeq} follows almost directly:
\begin{proposition}
	For every sentence $\phi \in \MSOP$ and every $k$, there exists
	a sentence $\widehat \phi \in \MSOW$ such that
	$L(\widehat \phi) = \Psip^{-1}(L(\phi)_{\le k})$.
\end{proposition}
\begin{proof}
	Let $L = \{w \in (\Omegaswb \setminus \{\id_\emptyset\})^+ \cap \Coh \mid \Psip(w) \models \phi \}$.
	By Lemma~\ref{lem:msodec}, we have $\widehat \phi' \in \MSOW$   such that $L = L(\widehat{\phi}')$.
	That is $L$ is a regular language of finite words.
	Thus $L' = (L \shuffle \id_\emptyset^\ast) \cap \Coh$
	and $L'' = L' \cup \id_\emptyset^\ast$ are effectively regular.
	(Here $L_1 \shuffle L_2$ denotes the shuffle of $L_1$ and $L_2$, that is, all interleavings of words in $L_1$ and $L_2$.)
	Finally, note that the question whether $\id_\emptyset \models \phi$ is decidable.
	We conclude by picking an $\MSOW$ sentence $\widehat{\phi}$ for $L'$ (if $\id_\emptyset \not \models \phi$) or $L''$ (if $\id_\emptyset \models \phi$) by the classical B\"uchi-Elgot-Trakhtenbrot theorem.
\end{proof}
As a consequence:
\begin{corollary}
	$\widehat \phi$ is $\sim$-invariant and $\widetildeL(\widehat \phi) = \Phi(L(\phi)_{\le k})$.
\end{corollary}

\subsection{From $\Coh_{\le k}$ to $\iiPoms_{\le k}$}\label{sec:HDAs-to-MSO}

In this section we prove the second assertion of Theorem~\ref{th:MSOeq}, claiming that there exists an $\MSOP$ sentence satisfied by the ipomsets obtained by gluing the sparse coherent words satisfying some $\MSOW$ sentence.

Our construction relies on the uniqueness of the sparse step decomposition (Lemma~\ref{le:ipomsparse}) $\bar\Phi(P)$, and the $\MSOP$-definability of the relations: ``event $x$ is started/terminated before event $y$ is started/terminated in $\bar\Phi(P)$" (Lemma~\ref{lem:compare} below).

More formally, let $P \in \iiPoms$ and $P_1 \ldots P_n = \bar\Phi (P)$.
Recall that this means that $P = P_1 \comp \cdots \comp P_n$ and $P_i$
alternates between starters and terminators.
Also recall the notation for starts and terminations of events of Section \ref{se:step-subsumptions}:
given $e \in P \setminus S_P$, $\xstart e$ is the (unique) step where $e$ is started in the decomposition.
For $e \in P \setminus T_P$, $\xend e$ is the (unique) step where $e$ is terminated.
For $x\in S_P$, $\xstart x = -\infty$, and for $x\in T_P$, $\xend x = +\infty$.

\label{pg:xstart}

\begin{example}
	\label{ex:stend}
	\quad
	Proceeding with Example~\ref{ex:msotohda},
	let
	$w_2 = P_1\dots P_6=
	\loset{\hphantom{\ibullet} a \ibullet \\ \ibullet c \ibullet} 
	\loset{\ibullet a \ibullet \\ \ibullet c \hphantom{\ibullet}}
	\loset{\ibullet a \ibullet \\ \hphantom{\ibullet} d \ibullet}
	\loset{\ibullet a \hphantom{\ibullet}\\ \ibullet d \ibullet}$
	$\loset{\hphantom{\ibullet} a \ibullet \\ \ibullet d \ibullet}
	\loset{ \ibullet a \\ \ibullet d } = \bar\Phi(P)$ (see also Example~\ref{ex:paths}).
	We have $\xstart {e_3} = -\infty, \xstart {e_1} = 1, \xstart {e_4} = 3$ and $\xstart {e_2} = 5$.
	Also, $\xend {e_3} = 2, \xend {e_1} = 4$ and $\xend {e_2} = \xend {e_4} = 6$.
	Further, $P_1$ contains $e_1$ since $\xstart {e_1} = 1$ and $e_3$ because $\xstart {e_3} \le 1 \le \xend {e_3}$;
	$P_4$ contains $e_1$ since $\xend {e_1} =4$ and $e_4$ because $\xstart {e_4} \le 4 \le \xend {e_4}$. 
\end{example}

Our construction relies on the following observation:

\begin{remark}\label{rem:keyObs}
  When $P_1 \ldots P_n = \bar\Phi (P)$ for some $P \in \iiPoms \setminus \Id$,
  then each starter $P_i$ contains precisely all~$e \in P$ such that $\xstart e \le i < \xend e$.
  That is, all events which are started before, at $P_i$ or never started, and are terminated after $P_i$ or never terminated. 
  In particular, $P_i$ starts all $e$ such that $\xstart e = i$.
  When it is a terminator, $P_i$ contains precisely all~$e \in P$ such that $\xstart e < i \le \xend e$, and terminates all $e$ such that $\xend e = i$.
  Note that $\xstart e<\xend e$ for all $e\in P$.
\end{remark}

These relations are $\MSOP$-definable. We first encode in the following lemma some of them:
\begin{lemma}
	\label{lem:compare}
	For $f,g \in \{\xst,\xen\}$ and ${\bowtie} \in \{{<}, {>}\}$,
	the relations $f(x) \bowtie g(y)$ are $\MSOP$-definable.
\end{lemma}

\begin{proof}
	We first define $\xend x < \xstart y$ as the formula $x < y$, together with
	$\xstart x < \xend y \eqdef \lnot (\xend y < \xstart x)$.
	Because starters and terminators alternate in the sparse step decomposition,
	we can then let
	\begin{align*}
		\xstart x < \xstart y
		& \eqdef  \big(\srci(x) \land \lnot \srci(y)\big) \lor
		\big(\exists z.\, \xstart x < \xend z \land \xend z < \xstart y\big) \\
		\xend x < \xend y
		& \eqdef \big(\tgti (y) \land \lnot \tgti (x)\big) \lor  \big(\exists z.\, \xend x < \xstart z \land \xstart z < \xend y\big)
	\end{align*}
	where $\srci(x) \land \lnot \srci(y)$  notably holds when both $x$ and $y$ are minimal in an ipomset $P$, and $y$ is started while $x$ is a source.
	Similarly  $\tgti (y) \land \lnot \tgti (x)$ holds in particular when both $x$ and $y$ are maximal, $x$ is terminated while $y$ is a target. 
\end{proof}


Observe that $\xend x < \xstart y$ implies $\lnot \tgti(x) \land \lnot \srci(y)$,
given that the end of the $x$-event precedes the beginning of the $y$-event.
As a consequence $\xstart x < \xstart y$ implies $\lnot \srci(y)$.
On the other hand $\xstart x < \xend y$  holds in particular if $\srci(x)$ or $\tgti(y)$
holds.

\begin{example}
	Continuing Example~\ref{ex:stend},
	observe that $P \models \xstart {e_3} <  \xstart {e_1}$ even if $e_3$ and $e_1$ are minimal and both occur in $P_1$.
	We have also that $P \not\models \xstart {e_1} <  \xstart {e_3}$, $P \models \xstart {e_1} <  \xstart {e_2}$, $P \not\models \xend {e_2} <  \xend {e_4}$ and $P \models \xend {e_3} <  \xend {e_4}$.	
	Observe that $P\models \xstart {e_1} < \xend {e}$ for all $e \in P$.
\end{example}

Then, using Lemma~\ref{lem:compare}, one can encode each step of $\Phip(P)$ with $\MSOP$: 
\begin{lemma}\label{lem:StTe(x)Def}
	For all $D \in \Omegaswb \setminus \Id$, there are $\MSOP$-formulas $D_\xst(x)$ and $D_\xen(x)$
	such that for all $P \in \iiPoms_{\le k}$ with $P_1 \ldots P_n = \Phip (P)$
	and valuations $\nu(x) = e \in P$, we have
	\begin{itemize}
		\item $P \models_\nu D_\xst(x)$ if and only if $\xstart e \neq - \infty$ and $P_{\xstart e} = D$;
		\item $P \models_\nu D_\xen(x)$ if and only if $\xend e \neq + \infty$ and $P_{\xend e} = D$.
	\end{itemize}
\end{lemma}

\begin{proof}
	We give the definition of $D_\xst(x)$, the one for $D_\xen(x)$
	is similar.
	We obviously set $D_\xst(x) \eqdef \bot$ when $D \notin \St_+$.
	Now assume $D \in \St_+$.
	Let $d_1 \evord \cdots \evord d_\ell$ the events in~$D$, and $a_i$
	the label of $d_i$.
	We first define a formula $\phi(x,x_1,\ldots,x_\ell)$ which is true when
	$x_1,\ldots,x_\ell$ are precisely the events occurring in the $\xstart x^{\text{th}}$ step of $\Phip(P)$ (see Remark~\ref{rem:keyObs}):
	\begin{multline*}
          \qquad
          \phi(x,x_1,\ldots,x_\ell) \eqdef {} \\
          \lnot \srci(x) \land \forall y.
          \Big( \lnot (\xstart x < \xstart y) \land (\xstart x < \xend y) \Big) \iff
          \bigvee_{1 \le i \le \ell} y = x_i
          .
          \qquad
	\end{multline*}
	Note that the condition $\lnot \srci(x)$ is here to ensure that $\xstart x \neq -\infty$.
	Note also that, given that $D$ is a starter, for any $i$, $d_i \in T_D$.
	We also need to ensure that sources of $D$ are preserved, \ie to know among $x_1,\dots,x_\ell$ which ones are started strictly before $x$.
	For all $i \in \{1,\ldots,\ell\}$, we let
	\[
	\psi_i(x,y) \eqdef
	\begin{cases}
	\xstart y < \xstart x & \text{if } d_i \in S_D \\
	\lnot (\xstart y < \xstart x )& \text{otherwise }
	\end{cases}
	\]
	Then, we use these formulas to specify which events should be started 
	together with $x$ in $D$ or not.
	We then define
	\[
	D_\xst(x) \eqdef
	\exists x_1, \ldots, x_\ell.\, 
	\phi(x,x_1,\ldots,x_\ell) \land
	\bigwedge_{1 \le i \le \ell} a_i(x_i) \land \psi_i(x,x_i) \land
	\bigwedge_{1 \le i < \ell} x_i \evord x_{i+1}.
	\] 
\end{proof}

The following lemma proves the second part of Theorem~\ref{th:MSOeq} when identities are not taken into account.

\begin{lemma}\label{prop:WtoPlabels}
	For every sentence $\phi \in \MSOW$,
	there exists a sentence $\overline \phi \in \MSOP$ 
	such that $L(\overline \phi) \cap \iiPoms_{\le k} \setminus \Id= \Psip(L(\phi) \cap \SCoh \setminus \Id)$.
\end{lemma}

\begin{proof}
	The key idea is that, given a sparse step decomposition
	$w = P_1 \dots P_n \in \Omegaswb^+ \setminus \Id$, every $P_i$ within $w$ contains some event $e$ which is either started or terminated in $P = \Psip(w)$.
	Thus, we associate to every position $i \in \{1,\ldots ,n\}$  a pair $(e,b)$ where $e \in P$ and $b \in \{\xst,\xen\}$ indicates whether we are looking for a starter or a terminator.
	The other events of $P$ are captured by the formulas of Lemma~\ref{lem:StTe(x)Def} and the other conditions are obtained inductively.

	More formally, we proceed similarly to the translation in the previous section.
	For every function $\tau$ from free first-order variables to $\{\xst,\xen\}$ 
	and every $\MSOW$ formula $\phi$, 
	we define $\overline \phi_\tau$ as follows:
	\begin{align*}
		\overline{\exists x.\, \phi}_\tau
		& := (\exists x.\, \lnot \srci(x) \land \overline{\phi}_{\tau[x \mapsto \xst]}) \lor 
		(\exists x.\, \lnot \tgti(x) \land \overline{\phi}_{\tau[x \mapsto \xen]}) \\
		\overline {x < y}_\tau
		& := \tau(x)(x) < \tau(y)(y) \quad \text{(see Lemma~\ref{lem:compare})} \\
		\overline {P(x)}_\tau
		& := P_{\tau(x)}(x) \quad \text{(see Lemma~\ref{lem:StTe(x)Def})} \\
		\overline{\exists X.\, \phi}_\tau
		& := \exists X_\xst, X_\xen. \, (\forall x.\, x \in X_\xst \implies \lnot \srci(x)) \land (\forall x.\, x \in X_\xen \implies \lnot \tgti(x)) \land \overline{\phi}_\tau \\
		\overline{x \in X}_\tau & := x \in X_{\tau(x)} \\
		\overline{\phi \lor \psi}_\tau & := \overline{\phi}_\tau \lor \overline{\psi}_\tau \\
		\overline{\lnot \phi}_\tau & := \lnot \overline{\phi}_\tau.
	\end{align*}
\end{proof}

Again, the second assertion of Theorem~\ref{th:MSOeq} follows almost directly:
\begin{proposition}
	For every sentence $\phi \in \MSOW$,
	there exists a sentence $\overline \phi \in \MSOP$
	such that $L(\overline \phi) = \Psip(L(\phi) \cap \SCoh)$.
\end{proposition}
\begin{proof}
	Recall that the sparse step decomposition of an identity is the identity itself.
	In addition, it is decidable whether an identity satisfies $\phi$ and there are finitely many identities in $\Omegaswb$.
	Moreover, identities are trivially $\MSOP$-definable.
	Let $\psi_1$ be the disjunction of the $\MSOP$ formulas that hold for the identities satisfying $\phi$, and $\psi_2$ the formula obtained from Lemma~\ref{prop:WtoPlabels}.
	Notice that Lemma~\ref{prop:WtoPlabels} doesn't specify what $\psi_2$ evaluates to on identities or on ipomsets of width $> k$, so it could be satisfied by ipomsets that we want to exclude.
	However, the set $\iiPoms_{\le k} \setminus \Id$ is easily definable in $\MSOP$ by
	\[
		\psi_3 \eqdef (\exists x. \lnot \srci(x) \lor \lnot \tgti(x)) \land \forall x_1,\ldots, x_{k+1}.\, 
		\bigvee_{1 \le i,j \le k+1} x_i < x_j.
	\]
	We then let $\overline \phi = \psi_1 \lor (\psi_2 \land \psi_3)$.
\end{proof}
In addition, as an immediate corollary, we have:
\begin{corollary}
	For every $\sim$-invariant sentence $\phi \in \MSOW$, 
	$L(\overline \phi) = \Psip(L(\phi) \cap \Coh) = \Psi(\widetilde L(\phi))$.
\end{corollary}

\subsection{A B\"uchi-Elgot-\!Trakhtenbrot theorem}
\label{sec:buchi}

Recall that  languages of HDAs have bounded width and are closed
under subsumption, unlike $\MSOP$-definable languages (see Example~\ref{ex:notwb}).
Therefore, we can only translate an $\MSOP$ formula into
an equivalent HDA if it has these two properties.
We have the following.

\begin{theorem}
	\label{th:main}
	Let $L \subseteq \iiPoms$.
	\begin{enumerate}
        \item If $L$ is $\MSOP$-definable, then $L_{\le k}\down$ is regular for all $k \in \Nat$.
        \item If $L$ is regular, then it is $\MSOP$-definable. 
	\end{enumerate}
	Moreover, the constructions are effective in both directions.
\end{theorem}

\begin{proof}
	Assume that $L$ is $\MSOP$-definable.
	Then by Theorem~\ref{th:MSOeq}, $\Psi^{-1}(L_{\le k})$ is also 
	$\MSOW$-definable.
	By the standard B\"uchi-Elgot-Trakhtenbrot and Kleene theorems,
	we can construct a rational expression $E$ over $\Omegaswb$ such
	that $L(E) = \Psi^{-1}(L_{\le k})$.
	By replacing concatenation of words by gluing composition in $E$ (see \cite[Proposition~21]{DBLP:journals/lmcs/FahrenbergJSZ24} or \cite[Lemmas~28-31]{DBLP:journals/corr/abs-2503-07881} for a detailed construction),
	we get that $L_{\le k}\down$ is rational and thus effectively regular by 
	Theorem~\ref{th:kleene}.
	
%

	Conversely, assume that $L$ is regular and let $L' = L \setminus \Id$ and $I = L \setminus L'$.
	Obviously $L'$ is also accepted by some HDA $\hda{H}$.
	Let $k$ be the dimension of $\hda{H}$.
	Note that $I \subseteq \Omegaswb$ is finite.
	Let $A$ be the ST-automaton built from $\hda{H}$. 
	Note that none of the initial states of $\hda{H}$ (hence $A$) are accepting.
	We have by Theorem~\ref{th:HDA-STA}, $L(A) = \Phi(L')$.	
	From $A$ we can build a classical automaton $B$ by just forgetting state labels such that $\Psip(L(B)) = L'$.
	In addition, one can easily build another classical automaton $C$ such that $L(C) = L(B) \cup I$.
	Thus $\Psip(L(C)) = L$.
	By construction, $\Phip(L) \subseteq L(C)$.
	Since $L(C)$ is regular, it is $\MSOW$-definable.
	By Theorem~\ref{th:MSOeq}, $\Psip(L(C) \cap \Phip(L)) = L$ is $\MSOP$-definable.
\end{proof}

As a special case, the following corollary holds for (downward-closed) languages of $\iiPoms_{\le k}$.

\begin{corollary}
	For all $k \in \Nat$, a language $L \subseteq \iiPoms_{\le k}$ is regular iff it is $\MSOP$-definable.
\end{corollary}

Since emptiness of HDAs is decidable \cite{AMRANE2025115156}, the following are also decidable:
\begin{enumerate}
\item Given $\phi \in \MSOP$ such that $L(\phi) = L(\phi)_{\le k}\down$,
  does there exist $P \in \iiPoms$ such that $P \models \phi$?
\item Given $\phi \in \MSOP$ such that $L(\phi) = L(\phi)_{\le k}\down$
  and an HDA $\hda{H}$, is $L(\hda{H}) \subseteq L(\phi)$?
\end{enumerate}
That is to say the following.

\begin{corollary}
	For $\MSOP$ sentences $\phi$ such that
	$L(\phi) = L(\phi)_{\le k}\down$, the satisfiability problem and the model-checking problem for HDAs  are both decidable.
\end{corollary}

Actually, looking more closely at our construction which
goes through ST-automata, we get the same result
for $\MSOP$ formulas even without the assumption that $L(\phi)$
is downward-closed (but still over $\iiPoms_{\le k}$, and not $\iiPoms$).
This could also be shown alternatively by observing that $\iiPoms_{\le k}$ has bounded treewidth (in fact, even bounded pathwidth), and applying 
Courcelle's theorem \cite{Courcelle90}. 
In fact our implied proof of decidability is relatively similar, using coherent words instead of path decompositions. 

Finally, using both directions of Theorem~\ref{th:main} and the closure properties of HDAs, we also get the following.

\begin{corollary}
	\label{cor:downMSOdef}
	For all $k\in \Nat$ and $\MSOP$-definable $L\subseteq \iiPoms_{\le k}$, $L\down$ is $\MSOP$-definable.
\end{corollary}

Note that this property does \emph{not} hold for the class of \emph{all} pomsets \cite{DBLP:conf/apn/FanchonM09}.
Indeed, \cite[Example 37]{DBLP:conf/apn/FanchonM09} exposes a pomset language $L$
of width $2$ (and not interval)
such that $L$ is MSO-definable but $L\down$ is not.
Whether there is a class strictly in-between interval (i)pomsets and general pomsets
for which Corollary \ref{cor:downMSOdef} remains true is a question we leave open.
	
\section{Conclusion}

In this paper we have explored interval pomsets with interfaces (ipomsets) from both an algebraic and a logical perspective. 
We have introduced two categorically equivalent definitions of ipomsets.
We have also shown that to every ipomset
corresponds an equivalence class of words, called step sequences.
This implies that interval ipomsets are freely generated by certain discrete ipomsets
(starters and terminators) under the relation which composes subsequent starters and subsequent terminators.
We have transferred this isomorphism to cover subsumption by characterizing subsumption of ipomsets in terms of the swapping of starters and terminators in step sequences.
Finally, this isomorphism also holds operationally:
we have demonstrated that higher-dimensional automata (HDAs) can be translated into ST-automata, which accept the step sequences corresponding to the ipomsets of the original HDA.


We have also shown that the correspondence between step sequences and interval pomsets materializes logically when a width bound is fixed.
Indeed, we have shown using an MSO interpretation of interval pomsets into sparse step sequences and from step sequences into interval pomsets, that a set of step sequences is MSO-definable if and only if the elements of these equivalence classes are MSO-definable, if and only if their corresponding interval pomsets are MSO-definable (Corollary~\ref{cor:MSOdef}).
This induces a Büchi-Elgot-\!Trakhtenbrot theorem for HDAs.
The constructions use in particular the Kleene theorems for HDAs and words, and the Büchi-Elgot-\!Trakhtenbrot theorem for words.
As corollaries, the satisfiability and model checking problems for HDAs are both decidable.

Another corollary of our BET-like theorem is that, unlike for non-interval pomsets,
subsumption closures of MSO formulas over (interval) ipomsets are effectively computable;
that is, from a formula $\phi$ we may compute a formula $\phi\down$
such that an ipomset satisfies $\phi\down$ if and only if it is subsumed by one which satisfies $\phi$.
However, the construction of $\phi\down$ is not efficient, as the current workflow is to transform $\phi$ to an HDA and then back to get $\phi\down$.
We leave open the question whether the characterization of subsumptions by elementary operations on step sequences
(Theorem~\ref{th:subsumption})
may
lead to
more efficient constructions.

The decidability of the model checking problem has motivated further research into the expressive power of first-order logic over ipomsets.
An initial step in this direction appears in \cite{DBLP:journals/corr/abs-2410-12493}.
Along similar lines, another operational model has begun to receive attention: $\omega$-HDAs, that is, HDAs over infinite interval ipomsets \cite{DBLP:journals/corr/abs-2503-07881}.
In both areas, substantial work remains to be done.


\newcommand{\Afirst}[1]{#1} \newcommand{\afirst}[1]{#1}

\end{document}